\pdfoutput=1
\documentclass[a4paper, UKenglish, cleveref, autoref, thm-restate]{lipics-v2021}
\author{Bartosz Bednarczyk}
{Computational Logic Group, Technische Universit{\"a}t  Dresden, Germany \and 
Institute of Computer Science, University of Wroc\l aw, Poland}
{bartosz.bednarczyk@cs.uni.wroc.pl}
{https://orcid.org/0000-0002-8267-7554}
{supported by “Diamentowy. Grant” no. DI2017 006447.}

\author{Maja Or\l{}owska }
{Institute of Computer Science, University of Wroc\l aw, Poland}
{}
{}
{}

\author{Anna Pacanowska}
{Institute of Computer Science, University of Wroc\l aw, Poland}
{}
{}
{}

\author{Tony Tan}
{Department of Computer Science and Information Engineering, National Taiwan University, Taiwan}
{tonytan@csie.ntu.edu.tw}
{}
{supported by
Taiwan Ministry of Science and Technology under grant no. 109-2221-E-002-143-MY3
and National Taiwan University under grant no.~109L891808.}

\authorrunning{B. Bednarczyk and M. Or\l{}owska and A. Pacanowska and T. Tan}
\Copyright{Bartosz Bednarczyk and Maja Or\l{}owska and Anna Pacanowska and Tony Tan}


\relatedversion{Extended version of the paper~\cite{ArxivVersion} is available on \href{https://arxiv.org/abs/2106.15250}{arxiv.org/abs/2106.15250}. Results from the 3rd section of the paper already appeared in the bachelor thesis~\cite{PracaLic} of A. Pacanowska and M. Or\l{}owska, written under an informal supervision of B. Bednarczyk.}

\EventEditors{Miko{\l}aj Boja\'{n}czyk and Chandra Chekuri}
\EventNoEds{2}
\EventLongTitle{41st IARCS Annual Conference on Foundations of Software Technology and Theoretical Computer Science (FSTTCS 2021)}
\EventShortTitle{FSTTCS 2021}
\EventAcronym{FSTTCS}
\EventYear{2021}
\EventDate{December 15--17, 2021}
\EventLocation{Virtual Conference}
\EventLogo{}
\SeriesVolume{213}
\ArticleNo{5}

\usepackage{lmodern}

\usepackage[abbreviations,british]{foreign} 

\usepackage{etoolbox}

\usepackage{xcolor}
\definecolor{darkmidnightblue}{rgb}{0.0, 0.2, 0.4}
\definecolor{persianplum}{rgb}{0.44, 0.11, 0.11}

\hypersetup{
  colorlinks  = true, 
  citecolor   = persianplum, 
  urlcolor    = persianplum, 
  linkcolor   = persianplum
}

\usepackage{tikz}
\usepackage{todonotes}
\usepackage{xspace}

\usepackage{amssymb} 
\usepackage{amsmath}
\usepackage{amsthm}
\usepackage{mathtools}

\usetikzlibrary{shadows}
\tikzset{
diagonal fill/.style 2 args={fill=#2, path picture={
\fill[#1, sharp corners] (path picture bounding box.south west) -|
                         (path picture bounding box.north east) -- cycle;}},
reversed diagonal fill/.style 2 args={fill=#2, path picture={
\fill[#1, sharp corners] (path picture bounding box.north west) |- 
                         (path picture bounding box.south east) -- cycle;}}
}

\usetikzlibrary{hobby,backgrounds,calc,trees}

\usepackage{tikz-cd}
\usetikzlibrary{matrix}

\usepackage{caption}
\usepackage{subcaption}

\newtheorem{fact}{Fact}

\newcommand{\complexityclass}[1]{\textsc{#1}} 

\newcommand{\ExpTime}{\complexityclass{ExpTime}} 
\newcommand{\ThreeNExpTime}{\complexityclass{3NExpTime}} 

\newcommand{\FO}{\mathrm{FO}}
\newcommand{\FOt}{\FO^2}
\newcommand{\Ct}{\mathrm{C}^2}

\newcommand{\FOtpercloc}{\FOt_{\textit{loc}\%}}
\newcommand{\FOtpercgl}{\FOt_{\textit{gl}\%}}

\newcommand{\GF}{\mathrm{GF}}
\newcommand{\GFt}{\GF^2}

\newcommand{\GFtpercloc}{\GF^2_{\textit{loc}\%}}
\newcommand{\GFtpercgl}{\GF^2_{\textit{gl}\%}}
\newcommand{\GFtpres}{\GF^2_{\textit{pres}}}

\newcommand{\GFkpercloc}{\GF^k_{\textit{loc}\%}}
\newcommand{\GFkpercgl}{\GF^k_{\textit{gl}\%}}
\newcommand{\GFkpres}{\GF^k_{\textit{pres}}}
\newcommand{\GFpres}{\GF_{\textit{pres}}}

\newcommand{\GFpercloc}{\GF_{\textit{loc}\%}}
\newcommand{\GFpercgl}{\GF_{\textit{gl}\%}}

\newcommand{\DL}[1]{\ensuremath{\mathcal{#1}}}   
\newcommand{\ALCI}{\DL{ALCI}}                    
\newcommand{\ALCIHbself}{\DL{ALCIH}b^{\textsf{self}}}     
\newcommand{\ALCISCCplusplus}{\DL{ALCISCC}^{++}} 
\newcommand{\ALCSCC}{\DL{ALCSCC}}                    

\newcommand{\existsperc}[2]{\exists^{{#1}{#2}\%}}
\newcommand{\existspercwithrel}[3]{\exists^{{#1}{#2}\%}_{#3}}
\newcommand{\Maj}{\existsperc{=}{50}}

\newcommand{\card}[3]{|#1|_{#2}^{#3}}

\newcommand{\SHTP}{\ensuremath{\mathsf{SHTP}}}

\newcommand{\str}[1]{{\mathcal{#1}}}

\newcommand{\cA}{\mathcal{A}}

\newcommand{\N}{{\mathbb{N}}}
\newcommand{\Z}{{\mathbb{Z}}} 
 
\newcommand{\Var}{\mathrm{Var}}

\newcommand{\phiequality}{\varphi_{\text{eq}}}
\newcommand{\phihalves}{\varphi_{\text{halves}}}
\newcommand{\phientry}{\varphi_{\text{entry}}^{\varepsilon_i}}
\newcommand{\phipartition}{\varphi_{\text{parti}}}
\newcommand{\phiaddition}{\varphi_{\text{add}}}
\newcommand{\phimultiplication}{\varphi_{\text{mult}}}
\newcommand{\philink}{\varphi_{\text{link}}}
\newcommand{\phireduction}[1]{\varphi_{\text{red}}^{#1}}
\newcommand{\phicount}{\varphi_{\text{count}}}
\newcommand{\phibfunc}{\varphi_{\text{bfunc}}}
\newcommand{\phivar}{\varphi_{\text{var}}^{\varepsilon}}
\newcommand{\phiuequalone}{\varphi_{u_i{=}1}}
\newcommand{\Avar}[1]{A_{#1}}
\newcommand{\FHalf}[1]{\textrm{FHalf}^{[#1]}}
\newcommand{\SHalf}[1]{\textrm{SHalf}^{[#1]}}
\newcommand{\Mult}[2]{\textrm{Mult}_{#1}{#2}}

\newcommand{\predHalf}{\textrm{Half}}
\newcommand{\predH}{\textrm{H}}
\newcommand{\predR}{\textrm{R}}
\newcommand{\predJ}{\textrm{J}}
\newcommand{\predU}{\textrm{U}}

\newcommand{\cN}{\mathcal{N}}

\newcommand{\vc}{\bar{c}}
\newcommand{\vu}{\bar{u}}
\newcommand{\vv}{\bar{v}}
\newcommand{\vx}{\bar{x}}

\newcommand{\offset}[1]{\text{offset}(#1)}
\newcommand{\prd}[1]{\text{prd}(#1)}

\newcommand{\etanull}{\eta_{\text{null}}}

\nolinenumbers
\title{On Classical Decidable Logics extended with Percentage Quantifiers and Arithmetics}
\titlerunning{On Classical Decidable Logics extended with Percentage Quantifiers and Arithmetics}

\begin{document}
\maketitle

\begin{abstract}
During the last decades, a lot of effort was put into identifying decidable fragments of first-order logic. 
Such efforts gave birth, among the others, to the two-variable fragment and the guarded fragment, 
depending on the type of restriction imposed on formulae from the language. 
Despite the success of the mentioned logics in areas like formal verification and knowledge representation, 
such first-order fragments are too weak to express even the simplest statistical constraints, 
required for modelling of influence networks or in statistical reasoning.

In this work we investigate the extensions of these classical decidable logics with percentage quantifiers, 
specifying how frequently a formula is satisfied in the indented model. 
We show, surprisingly, that all the mentioned decidable fragments become undecidable under such extension, 
sharpening the existing results in the literature. 
Our negative results are supplemented by decidability of the two-variable guarded fragment 
with even more expressive counting, namely Presburger constraints.
Our results can be applied to infer decidability of various modal and description logics, 
\eg Presburger Modal Logics with Converse or $\ALCI$, with expressive cardinality constraints.
\end{abstract}

\section{Introduction} \label{sec:intro}

Since the works of Church, Turing and Trakhtenbrot,
it is well-known that the (finite) satisfiability and validity problems for the First-Order Logic ($\FO$) are undecidable~\cite{trakhtenbrot}.
Such results motivated researchers to study restricted classes of $\FO$ that come with decidable
satisfiability problem, such as the prefix classes~\cite{BorgerGG1997}, 
fragments with fixed number of variables~\cite{Scott62}, 
restricted forms of quantification~\cite{AndrekaNB98,Quine76} and 
the restricted use of negation~\cite{BaranyCS15}. 
These fragments have found many applications in the areas of 
knowledge representation, automated reasoning and  program verification, just to name a few.
To the best of our knowledge, none of the known decidable logics incorporate 
a feature that allows for stating even a very modest statistical property.
For example, one may want to state that
``to qualify to be a major, one must have at least 51\% of the total votes'', 
which may be useful to formalise, \eg the voting~systems.

\textbf{Our results.}
In this paper, we revisit the satisfiability problem for some of the most prominent fragments of $\FO$, 
namely the two-variable fragment $\FOt$ and the guarded fragment~$\GF$.
We extend them with the so-called \emph{percentage} quantifiers, 
in two versions: \emph{local} and \emph{global}.
Global percentage quantifiers are quantifiers of the form $\exists^{=q\%} x\ \varphi(x)$,
which states that the formula $\varphi(x)$ holds on exactly $q\%$ of the domain elements.
Their local counterparts are quantifiers of the form $\exists^{=q\%}_R y\ \varphi(x,y)$,
which intuitively means that exactly $q\%$ of the $R$-successors of an element $x$ satisfy $\varphi$.

In this paper, we show that both $\FOt$ and $\GF$ become undecidable 
when extended with percentage quantifiers of any type. 
In fact, the undecidability of $\GF$ already holds for its three variable fragment~$\GF^3$. 
Our results strengthen the existing undecidability proofs of~$\ALCISCCplusplus$ from~\cite{BaaderBR20} and 
of~$\FOt$ with equicardinality statements (implemented via the H{\"a}rtig quantifier) from~\cite{Gradel:undec} and 
contrast with the decidability of $\FOt$ with counting quantifiers ($\Ct$)~\cite{GradelOR97,PacholskiST00,Pratt-Hartmann05} and 
modulo and ultimately-periodic counting quantifiers~\cite{BenediktKT20}.

Additionally, we show that the decidability status of $\GF$ can be regained 
if we consider $\GF^2$, \ie the intersection of $\GF$ and $\FOt$, 
which is still a relevant fragment of $\FO$ that captures standard description logics up to $\ALCIHbself$~\cite{DLBook,Gradel:DL}.
We in fact show a stronger result here:
$\GF^2$ remains decidable when extended with \emph{local Presburger quantifiers},
which are essentially Presburger constraints on the neighbouring elements, 
\eg we can say that the number of red outgoing edges plus twice the number of blue outgoing edges 
is at least three times as many as the number of green incoming edges.

We stress here that the semantics of global percentage quantifiers makes sense only over finite domains and hence, 
we study the satisfiability problem over finite models only. 
Similarly, the semantics of local percentage quantifiers only makes sense
if the models are finitely-branching. 
While we stick again to the finite structures, 
our results on local percentage quantifiers also can be transferred to the case of (possibly infinite) finitely-branching structures.

\textbf{Related works.}
Some restricted fragments of $\GFt$ extended with arithmetics, 
namely the (multi) modal logics, were already studied in the literature~\cite{DemriL10,KupkeP10,Baader17,BaaderBR20}, 
where the decidability results for their finite and unrestricted satisfiability were obtained.
However, the logics considered there do not allow the use of the inverse of relations.
Since $\GFt$ captures the extensions of all the aforementioned logics with the inverse relations,
our decidability results subsume those in~\cite{DemriL10,KupkeP10,Baader17,BaaderBR20}.
We note that prior to our paper, it was an open question whether any of these decidability results
still hold when inverse relations are allowed~\cite{BaaderBR20}.
In our approach, despite the obvious difference in expressive power, 
we show that $\GFt$ with Presburger quantifiers can be encoded directly into the two-variable logic with counting quantifiers~\cite{GradelOR97,PacholskiST00,Pratt-Hartmann05}, which we believe is relatively simple and avoids cumbersome reductions of the satisfiability problem into integer programming. 


\section{Preliminaries}

We employ the standard terminology from finite model theory~\cite{Libkin04}.
We refer to structures/models with calligraphic letters $\str{A}, \str{B}, \str{M}$ and 
to their universes with the corresponding capital letters $A,B,M$. 
We work only on structures with \emph{finite} universes over purely relational (\ie constant- and function-free) signatures of arity $\leq 2$ containing the equality predicate $=$.
We usually use $a, b, \ldots$ to denote elements of structures, 
$\bar{a}, \bar{b}, \ldots$ for tuples of elements, $x, y, \ldots$ for variables and
$\bar{x}, \bar{y}, \ldots$ for tuples of variables (all of these possibly with some decorations). 
We write $\varphi(\bar{x})$ to indicate that all free variables of~$\varphi$ are in $\bar{x}$.
We write $\str{M},x/a \models \varphi(x)$ to denote that $\varphi(x)$ holds in the structure $\str{M}$ when
the free variable $x$ is assigned with element $a$.
Its generalization to arbitrary number of free variables is defined similarly.
The (finite) satisfiability problem is to decide whether an input formula has a (finite) model.

\subsection{Percentage quantifiers}
For a formula $\varphi(x)$ with a single free-variable $x$,
we write $\card{\varphi(x)}{\str{M}}{}$ to denote the total number of elements of $\str{M}$ satisfying $\varphi(x)$.
Likewise, for an element $a \in M$ and a formula $\varphi(x,y)$ with free variables $x$ and $y$,
we write $\card{\varphi(x,y)}{\str{M}}{x/a}$ to denote the total number of elements 
$b \in M$ such that $(a,b)$ satisfies $\varphi(x,y)$.

The \emph{percentage quantifiers} are quantifiers of the form  
$\existsperc{=}{q}x \ \varphi(x,y) $, 
where $q$ is a rational number between $0$ and $100$, stating that 
exactly $q\%$ of domain elements satisfy $\varphi(x,y)$ with $y$ known upfront.
Formally:
\[
\str{M}, y/a \models \existsperc{=}{q}x \; \varphi(x,y) \qquad \text{iff} \qquad  \card{\varphi(x,y)}{\str{M}}{y/a} = \frac{q}{100} \cdot |M|.
\]
Percentage quantifiers for other thresholds (\eg for $<$) are defined analogously.
We stress here that the above quantifiers count \emph{globally}, 
\ie they take the whole universe of $\str{M}$ into account.
This motivates us to define their local counterpart, as follows: 
for a binary\footnote{Local percentage quantifiers for predicates of arity higher than two 
can also be defined but we will never use them. 
Hence, for simplicity, we define such quantifiers only for binary relations.} relation $R$ and a rational $q$ between $0$ and $100$,
we define the quantifier $\existspercwithrel{=}{q}{R}y\;\varphi(x, y)$, 
which evaluates to true whenever exactly $q$\% of $R$-successors $y$ of $x$ satisfy $\varphi(x,y)$.
Formally,
\[
\str{M},x/a \; \models\; \existspercwithrel{=}{q}{R}y \ \varphi(x, y) \qquad \text{iff} \qquad
\card{R(x,y) \land \varphi(x, y)}{\str{M}}{x/a} = \frac{q}{100} \cdot \card{R(x,y)}{\str{M}}{x/a}.
\]
We define the percentage quantifiers w.r.t. $R^-$ (\ie the inverse of $R$) and for other thresholds analogously.

\subsection{Local Presburger quantifiers}\label{subsec:local-presb-quantifiers}

The \emph{local Presburger quantifiers} are expressions of the following form:
\[
\sum_{i=1}^n\ \lambda_i\cdot \#_y^{r_i}[\varphi_i(x,y)] \quad \circledast \quad \delta
\]
where $\lambda_i$, $\delta$ are integers;
$r_i$ is either $R$ or $R^{-}$ for some binary relation $R$;
$\varphi_i(x,y)$ is a formula with free variables $x$ and $y$;
and $\circledast$ is one of $=$, $\neq$, $\leq$, $\geq$, $<$, $>$, $\equiv_d$ or $\not\equiv_d$,
where $d \in \N_+$.
Here $\equiv_d$ denotes the congruence modulo $d$.
Note that the above formula has one free variable~$x$.

Intuitively, the expression $\#_y^{r_i}[\varphi_i(x,y)]$ denotes the number of $y$'s that satisfy $r_i(x,y)\wedge\varphi_i(x,y)$ and evaluates to true on $x$, if the (in)equality $\circledast$ holds.
Formally, 
\[  \str{M},x/a   \quad \models \quad  \sum_{i=1}^n \lambda_i\cdot \#_y^{r_i}[\varphi_i(x,y)] \ \circledast \ \delta \quad \text{iff} \quad
\sum_{i=1}^n\ \ \lambda_i\cdot \card{r_i(x,y)\wedge \varphi_i(x,y)}{\str{M}}{x/a} \quad \circledast \quad \delta \]
Note that local percentage quantifiers can be expressed with Presburger quantifiers, \eg
$\exists^{50\%}_Ry \varphi(x,y)$
can be expressed as local Presburger quantifier:
$\#_y^R[\varphi(x,y)] - \frac{1}{2} \#_y^R[\top] = 0$.

\subsection{Logics}

In this paper we mostly consider two fragments of first-order logic, 
namely \emph{the two-variable fragment} $\FOt$ and \emph{the guarded fragment} $\GF$.
The former logic is a fragment of $\FO$ in which we can only use the variables $x$ and $y$.
By allowing local and global percentage quantifiers in addition to the standard universal and existential quantifiers, 
we obtain the logics $\FOtpercloc$ and $\FOtpercgl$.
The latter logic is defined by relativising quantifiers with relations. 
More formally, $\GF$ is the smallest set of first-order formulae such that the following holds.
\begin{itemize}
\item 
$\GF$ contains all atomic formulae $R(\bar{x})$ and equalities between variables.
\item 
$\GF$ is closed under boolean connectives.
\item 
If $\psi(\bar{x}, \bar{y})$ is in $\GF$ and $\gamma(\bar{x}, \bar{y})$ is a relational atom containing all free variables of $\psi$,
then both $\forall{\bar{y}} \; \gamma(\bar{x}, \bar{y}) \to \psi(\bar{x}, \bar{y})$ and 
$\exists{\bar{y}} \; \gamma(\bar{x}, \bar{y}) \land \psi(\bar{x}, \bar{y})$ are in $\GF$.
\end{itemize} 
By allowing global percentage quantifiers additionally in place of existential ones, we obtain the logic $\GFpercgl$. 
We obtain the logic $\GFpercloc$ by extending $\GF$'s definition with the rule:\footnote{Note that 
$R$ in the subscript of a quantifier serves the role of a ``guard''.}
\begin{itemize}
\item 
$\existspercwithrel{=}{q}{R}y \; \varphi(x, y)$ is in $\GFpercloc$ iff $\varphi(x,y)$ in $\GF$ with free variables $x,y$.
\end{itemize}
Similarly, we obtain $\GFpres$ by extending $\GF$'s definition with the rule:
\begin{itemize}
\item 
$\sum_{i=1}^n\ \lambda_i\cdot \#_y^{r_i}[\varphi_i(x,y)] \; \circledast \; \delta$ is in $\GFpres$ iff $\varphi_i(x,y)$ are in $\GF$ with free variables $x,y$.
\end{itemize}
Finally, we use $\GFkpercgl$, $\GFkpercloc$ and $\GFkpres$ to denote the $k$-variable fragments of the mentioned logics. 
Specifically, we use $\GFtpercgl$, $\GFtpercloc{}$ and $\GFtpres{}$ for the two-variable fragments.

\subsection{Semi-linear sets}

Since we will exploit the semi-linear characterization of Presburger constraints,
we introduce some terminology.
The term \emph{vector} always means \emph{row vectors}.
For vectors $\vv_0,\vv_1,\ldots,\vv_k \in \N^{\ell}$,
we write $L(\vv_0;\vv_1,\ldots,\vv_k)$ to denote the set:
\begin{eqnarray*}
L(\vv_0;\vv_1,\ldots,\vv_k) & := &
\left\{
\begin{array}{l|l}
\vu\in \N^{\ell} & \vu= \vv_0 + \sum_{i=1}^k n_i\vv_i\ \text{for some}\ n_1,\ldots,n_k \in \N
\end{array}
\right\}
\end{eqnarray*}
A set $S\subseteq \N^{\ell}$ is a \emph{linear} set, if 
$S = L(\vv_0;\vv_1,\ldots,\vv_k)$, for some $\vv_0,\vv_1,\ldots,\vv_k\in \N^{\ell}$.
In this case, the vector $\vv_0$ is called the {\em offset} vector of $S$,
and $\vv_1,\ldots,\vv_k$ are called the {\em period} vectors of $S$.
We denote by $\offset {S}$ the offset vector of $S$, \ie $\vv_0$
and $\prd {S}$ the set of period vectors of $S$, \ie $\{\vv_1,\ldots,\vv_k\}$.
A \emph{semilinear} set is a finite union of linear sets.

The following theorem is a well-known result by Ginsburg and Spanier~\cite{semilinear}
which states that every set $S\subseteq \N^{\ell}$ definable by Presburger formula is a semilinear set.
See~\cite{semilinear} for the formal definition of Presburger formula.

\begin{theorem}[\cite{semilinear}]\label{theo:semilinear}
For every Presburger formula $\varphi(x_1,\ldots,x_{\ell})$ 
with free variables $x_1,\ldots,x_{\ell}$,
the set $\{\vu\in \N^{\ell}\ |\ \varphi(\vu)\ \text{holds in}\ \N\}$ is semilinear.
Moreover, given the formula $\varphi(x_1,\ldots,x_{\ell})$,
one can effectively compute a set of tuples of vectors 
$
    \{(\vv_{1,0},\ldots,\vv_{1,k_1}),\ldots,(\vv_{p,0},\vv_{p,1},\ldots,\vv_{p,k_p})\}
$
such that $\{\vu\in \N^{\ell}\ |\ \varphi(\vu)\ \text{holds in}\ \N\}$ is equal to $\textstyle\bigcup_{i=1}^{p}\ L(\vv_{i,0};\vv_{i,1},\ldots,\vv_{i,k_i})
$.
\end{theorem}

\subsection{Types and neighbourhoods}
A \emph{$1$-type} over a signature $\Sigma$ is a maximally consistent set of unary predicates from $\Sigma$ or their negations,
where each atom uses only one variable $x$.
Similarly, a \emph{$2$-type} over $\Sigma$ is a maximally consistent set of binary predicates from $\Sigma$ or their negations
containing the atom $x\neq y$,
where each atom or its negation uses two variables $x$ and $y$.\footnote{We should remark here
that the standard definition of 2-type, such as in~\cite{GradelKV97,Pratt-Hartmann05},
a 2-type also contains unary predicates or its negation involving variable $x$ or $y$.
However, for our purpose, it is more convenient to define a 2-type as consisting of only
binary predicates that strictly use both variables $x$ and $y$.
Note also that we view a binary predicate such as $R(x,x)$ as a unary predicate.}

Note that $1$-types and $2$-types can be viewed as quantifier-free formulae that are the conjunction of their elements.
We will use the symbols $\pi$ and $\eta$ (possibly indexed) to denote 1-type and 2-type, respectively.
When viewed as formula, we write $\pi(x)$ and $\eta(x,y)$, respectively.
We write $\pi(y)$ to denote formula $\pi(x)$ with $x$ being substituted with $y$.
The $2$-type that contains only the negations of atomic predicates is called the \emph{null} type,
denoted by $\etanull$. Otherwise, it is called a \emph{non-null} type.

For a $\Sigma$-structure $\str{M}$, 
the \emph{type of an element} $a \in M$ is the unique $1$-type $\pi$ that $a$ satisfies in~$\str{M}$.
Similarly, the type of a pair $(a,b)\in M\times M$, where $a\neq b$, 
is the unique $2$-type that $(a,b)$ satisfies in $\str{M}$.
For an element $a\in M$, the {\em $\eta$-neighbourhood} of $a$,
denoted by $\cN_{\str{M},\eta}(a)$,
is the set of elements $b$ such that $\eta$ is the 2-type of $(a,b)$.
Formally, 
\[
\cN_{\str{M},\eta}(a) \; :=\;  
\left\{
\begin{array}{l|l}
b \in M & \str{M},x/a,y/b\models \eta(x,y)
\end{array}
\right\}.
\]
The {\em $\eta$-degree} of $a$, denoted by $\deg_{\str{M},\eta}(a)$, 
is the cardinality of $\cN_{\str{M},\eta}(a)$.

Let $\eta_1,\ldots,\eta_{\ell}$ be an enumeration of all the non-null types.
The {\em degree of $a$ in $\str{M}$} is defined as the vector $
\deg_{\str{M}}(a)\;  :=\;  
(\deg_{\str{M},\eta_1}(a), \cdots ,\deg_{\str{M},\eta_{\ell}}(a))
$.
Intuitively, $\deg_{\str{M}}(a)$ counts the number of elements adjacent to $a$
with non-null type.
We note that our logic can be easily extended with atomic predicates of the form
of a linear constraint $C$ over the variables $\deg_{\eta}(x)$'s or  
$\deg(x)\in S$, where $S$ is a semilinear set.
Semantically, $\str{M},x/a \models C$ iff the linear constraint $C$
evaluates to true when each $\deg_{\eta}(x)$ is substituted with $\deg_{\str{M},\eta}(a)$
and $\str{M},x/a\models \deg(x)\in S$ iff $\deg_{\str{M}}(a)\in S$.
We stress that these atomic predicates will only be used to facilitate the proof of our decidability result.


\section{Negative results}\label{sec:negative_results}

In this section we turn our attention to the negative results announced in the introduction.

\subsection{Two-Variable Fragment}

We start by proving that the two-variable fragment of $\FO$ extended with percentage quantification has undecidable finite satisfiability problem. Actually, in our proof, we will only use the $\Maj$ quantifier.  
Our results strengthen the existing undecidability proofs of~$\ALCISCCplusplus$ from~\cite{BaaderBR20} and of~$\FOt$ with equicardinality statements (implemented via the H{\"a}rtig quantifier) from~\cite{Gradel:undec}. 
Roughly speaking, our counting mechanism is weaker: we cannot write arbitrary Presburger constraints (as it is done in~\cite{BaaderBR20}) nor compare sizes of any two sets (as it is done in~\cite{Gradel:undec}). 
Nevertheless, we will see that in our framework we can express ``functionality'' of a binary relation and ``compare'' cardinalities of sets, but under some technical assumptions of dividing the intended models into halves. 
Due to such technicality, we cannot simply encode the undecidability proofs of~\cite{BaaderBR20,Gradel:undec} and we need to prepare our proof ``from scratch''.  

Our proof relies on encoding of Hilbert's tenth problem, whose simplified version is introduced below.
In the classical version of \emph{Hilbert's tenth problem} we ask whether a \emph{diophantine equation}, \ie a polynomial equation with integer coefficients, has a solution over~$\N$. 
It is well-known that such problem is undecidable~\cite{Matiyasevich}.
By employing some routine transformations (\eg by rearranging terms with negative coefficients, by replacing exponentiation by multiplication and by introducing fresh variable for partial results of multiplications or addition), one can reduce any diophantine equation to an equi-solvable system of equations, where the only allowed operations are addition or multiplication of two variables or assigning the value one to some of them.
We refer to the problem of checking solvability (over~$\N$) of such systems of equations as~$\SHTP$ (\emph{simpler Hilbert's tenth problem}) and present its precise definition next.
Note that, by the described reduction, $\SHTP$ is undecidable.

\begin{definition}[$\SHTP$]
An input of~$\SHTP$ is a system of equations~$\varepsilon$, where each of its entries~$\varepsilon_i$ is in one of the following forms: 
(i)~$u_i = 1$, 
(ii)~$u_i = v_i + w_i$, 
(iii)~$u_i = v_i \cdot w_i$, 
where~$u_i,v_i,w_i$ are pairwise distinct \emph{variables} from some countably infinite set~$\Var$. 
In~$\SHTP$ we ask whether an input system of equations~$\varepsilon$, as described before, has a solution over~$\N$. 
\end{definition}


\subsubsection{Playing with percentage quantifiers}\label{subsec:playing-with}
Before reducing $\SHTP$ to $\FOtpercgl$, let us gain more intuitions of $\FOtpercgl$ and introduce a useful trick employing percentage quantifiers to express equi-cardinality statements.
Let $\str{M}$ be a finite structure and let $\predHalf, \predR, \predJ$ be unary predicates.
We say that $\str{M}$ is \emph{$(\predHalf, \predR, \predJ)$-separated}  whenever it satisfies the following conditions:
(a) exactly half of the domain elements from $\str{M}$ satisfy $\predHalf$
(b) the satisfaction of $\predR$ implies the satisfaction of $\predHalf$
(c) the satisfaction of $\predJ$ implies the non-satisfaction of $\predHalf$.
Roughly speaking, the above conditions entail that the elements satisfying $\predR$ and 
those satisfying~$\predJ$ are in different halves of the model.
We show that under these assumptions one can enforce the equality $\card{\predR(x)}{\str{M}}{} = \card{\predJ(x)}{\str{M}}{}$.
Indeed, such a property can be expressed in $\FOtpercgl$ with the following formula $\phiequality(\predHalf, \predR, \predJ)$:

\begin{center}
    \tikzset{every picture/.style={line width=0.75pt}} 

    \begin{tikzpicture}[x=0.75pt,y=0.75pt,yscale=-1,xscale=1]

    \draw  [fill={rgb, 255:red, 126; green, 211; blue, 33 }  ,fill opacity=0.61 ] (141,29.33) -- (207,29.33) -- (207,98) -- (141,98) -- cycle ;
    \draw   (207,29.33) -- (273,29.33) -- (273,98) -- (207,98) -- cycle ;
    \draw   (154.75,68.17) .. controls (154.75,58.55) and (163.37,50.75) .. (174,50.75) .. controls (184.63,50.75) and (193.25,58.55) .. (193.25,68.17) .. controls (193.25,77.79) and (184.63,85.58) .. (174,85.58) .. controls (163.37,85.58) and (154.75,77.79) .. (154.75,68.17) -- cycle ;
    \draw  [fill={rgb, 255:red, 208; green, 2; blue, 27 }  ,fill opacity=0.65 ] (220.75,68.17) .. controls (220.75,58.55) and (229.37,50.75) .. (240,50.75) .. controls (250.63,50.75) and (259.25,58.55) .. (259.25,68.17) .. controls (259.25,77.79) and (250.63,85.58) .. (240,85.58) .. controls (229.37,85.58) and (220.75,77.79) .. (220.75,68.17) -- cycle ;

    \draw (105,56) node [anchor=north west][inner sep=0.75pt]    {$\str{A} :=$};

    \draw (284,56) node [anchor=north west][inner sep=0.75pt]    {$\models \phiequality(\predHalf, \predR, \predJ) := \Maj{x}\; (\predHalf(x) \wedge \neg \predR(x)) \vee \predJ(x)$};
    \draw (160,32) node [anchor=north west][inner sep=0.75pt]    {$\predHalf$};
    \draw (220,32) node [anchor=north west][inner sep=0.75pt]    {$\neg \predHalf$};
    \draw (168,60.5) node [anchor=north west][inner sep=0.75pt]    {$\predR$};
    \draw (237,60.5) node [anchor=north west][inner sep=0.75pt]    {$\predJ$};

    \end{tikzpicture}
\end{center}

For intuitions on $\phiequality(\predHalf, \predR, \predJ)$, consult the above picture. 
We basically take all the elements satisfying $\predHalf$ (so exactly half of the domain elements, indicated by the green area). 
Next, we discard the elements labelled with $\predR$ (so we get the green area without the circle inside) and replace them with the elements satisfying $\predJ$ (the red circle, note that $\predJ^{\str{A}}$ and $\predR^{\str{A}}$ are disjoint!). 
The total number of selected elements is equal to half of the domain, thus $|\predJ^{\str{M}}| = |\predR^{\str{M}}|$.
The following fact is a direct consequence of the semantics of $\FOtpercgl$.
\begin{fact}\label{fact:equality}
For $(\predHalf, \predR, \predJ)$-separated $\str{M}$ we have $\str{M} \models \phiequality(\predHalf, \predR, \predJ)$ iff $\card{\predR(x)}{\str{M}}{} = \card{\predJ(x)}{\str{M}}{}$.
\end{fact}

\subsubsection{Undecidability proof}
Until the end of this section, let us fix $\varepsilon$, a valid input of $\SHTP$. 
By $\Var(\varepsilon) = \{ u, v, w, \ldots \}$ we denote the set of all variables appearing in $\varepsilon$, and with~$|\varepsilon|$ we denote the total number of entries in $\varepsilon$.
Let $\str{M}$ be a finite structure.

The main idea of the encoding is fairly simple: in the intended model $\str{M}$ some elements will be labelled with $\Avar{u}$ predicates, ranging over variables $u \in \Var(\varepsilon)$, and the number of such elements will indicate the value of $u$ in an example solution to $\varepsilon$. 
The only tricky part here is to encode multiplication of variables. 
Once $\varepsilon$ contains an entry $w = u \cdot v$, we need to ensure that $\card{\Avar{w}(x)}{\str{M}}{} = \card{\Avar{u}(x)}{\str{M}}{} \cdot \card{\Avar{v}(x)}{\str{M}}{}$ holds. 
It is achieved by linking, via a binary relation $\Mult{}{}^{\str{M}}$, each element from $\Avar{u}^{\str{M}}$ with exactly~$\card{\Avar{v}(x)}{\str{M}}{}$ elements satisfying~$\Avar{w}$, which relies on imposing equicardinality statements. 
To ensure that the performed multiplication is correct, each element labelled with~$\Avar{w}^{\str{M}}$ has exactly one predecessor from~$\Avar{u}^{\str{M}}$ and hence the relation~$\Mult{}{}^{\str{M}}$ is backward-functional.

We start with a formula inducing a labelling of elements with variable predicates and ensuring that all elements of~$\str{M}$ satisfy at most one variable predicate. 
Note that it can happen that there will be auxiliary elements that are not labelled with any of the variable predicates.

\begin{center}
\begin{description} \label{form:phivar}
\itemsep 0 cm 
\item[($\phivar$)] \; $\forall{x} \bigwedge_{u \neq v \in \Var(\varepsilon)} 
\neg (\Avar{u}(x) \wedge \Avar{v}(x) )$.
\end{description}
\end{center}

%
We now focus on encoding the entries of~$\varepsilon$. 
For an entry $\varepsilon_i$ of the form~$u_i=1$ we write:
\begin{center}
\begin{description} \label{form:phiuequalone}
\itemsep 0 cm 
\item[($\phiuequalone$)] \;
    $
     \exists{x} \; \Avar{u_i}(x) \wedge \forall{x}\forall{y} \; 
     (\Avar{u_i}(x) \wedge \Avar{u_i}(y)) \rightarrow x=y
    $
\end{description}
\end{center}
\begin{fact} \label{fact:phiuequalone}
$\str{M} \models \phiuequalone$ holds iff there is exactly one element in~$\str{M}$ satisfying~$\Avar{u}(x)$.
\end{fact}

\noindent To deal with entries $\varepsilon_i$ of the form~$w_i = u_i+v_i$ or $w_i = u_i \cdot v_i$ we need to ``prepare an area'' for the encoding, similarly to~\cref{subsec:playing-with}.  First, we cover domain elements of $\str{M}$ by \emph{layers}. 
The $i$-th layer is divided into halves with $\FHalf{i}$ and $\SHalf{i}$ predicates with:
\begin{center}
\begin{description} \label{form:phihalves}
\itemsep 0 cm 
\item[($\phihalves^{i}$)] \;
        $
        \forall{x} \left( \FHalf{i}(x) \leftrightarrow \neg \SHalf{i}(x) \right)
        \wedge \Maj{x} \ \FHalf{i}(x)
        $
\end{description}
\end{center}
\begin{fact} \label{fact:phihalves}
$\str{M} \models \phihalves^{i}$ holds iff exactly half of the domain elements from~$\str{M}$ are labelled with~$\FHalf{i}$ and the other half of elements are labelled with~$\SHalf{i}$.
\end{fact}

Second, we need to ensure that in the $i$-th layer of $\str{M}$, the elements satisfying $\Avar{u_i}$ or $\Avar{v_i}$ are in the first half, whereas elements satisfying $\Avar{w_i}$ are in the second half. 
We do it with:
\begin{center}
\begin{description} \label{form:phipartition}
\itemsep 0 cm 
\item[($\phipartition^{i}(u_i,v_i,w_i)$)] \;
        $
        \forall{x} \left( 
        [(\Avar{u_i}(x) {\vee} \Avar{v_i}(x)) \rightarrow \FHalf{i}(x)] 
        \wedge [\Avar{w_i}(x) \rightarrow \SHalf{i}(x)] 
        \right)
        $
\end{description}
\end{center}
\begin{fact} \label{fact:phipartition}
$\str{M} \models \phipartition^{i}(u_i,v_i,w_i)$ holds iff for all elements~$a \in M$, if $a$ satisfies $\Avar{u_i}(x) {\vee} \Avar{v_i}(x)$
then $a$ also satisfies $\FHalf{i}(x)$ and if $a$ satisfies $\Avar{w_i}(x)$ then $a$ also satisfies $\SHalf{i}(x)$.
\end{fact}

Gathering the presented formulae, we call a structure~$\str{M}$ \emph{well-prepared}, if it satisfies the conjunction of all previous formulae over~$1 \leq i \leq |\varepsilon|$ and over all entries~$\varepsilon_i$ from the system~$\varepsilon$. 
The forthcoming encodings will be given under the assumption of \emph{well-preparedness}.

Now, for the encoding of addition, assume that~$\varepsilon_i$ is of the form~$u_i + v_i = w_i$.
Thus in our encoding, we would like to express that $\card{\Avar{u_i}(x)}{\str{M}}{} + \card{\Avar{v_i}(x)}{\str{M}}{}  = \card{\Avar{w_i}(x)}{\str{M}}{}$, which is clearly equivalent to $\card{\Avar{w_i}(x)}{\str{M}}{} - \card{\Avar{u_i}(x)}{\str{M}}{} - \card{\Avar{v_i}(x)}{\str{M}}{} = 0$ and also to $\card{\Avar{w_i}(x)}{\str{M}}{} + \card{\FHalf{i}(x)}{\str{M}}{} - \card{\Avar{u_i}(x)}{\str{M}}{} - \card{\Avar{v_i}(x)}{\str{M}}{} = \card{\FHalf{i}(x)}{\str{M}}{}$.
Knowing that exactly~$50\%$ of domain elements of an intended model satisfy $\FHalf{i}$ and that $\Avar{u_i}, \Avar{v_i}$ and $\Avar{w_i}$ label disjoint parts of the model, we can write the obtained equation as an~$\FOtpercgl$ formula:
\begin{center}
\begin{description} \label{form:addition}
\itemsep 0 cm 
\item[($\phiaddition^{i}(u_i,v_i,w_i)$)] \;
        $
    \Maj{x} \left( 
            \Avar{w_i}(x) \vee (\FHalf{i}(x) \wedge  {\neg}\Avar{u_i}(x) \wedge {\neg}\Avar{v_i}(x)) 
    \right)
        $
\end{description}
\end{center}
Note that the above formula is exactly the $\phiequality(\predHalf, \predR, \predJ)$ formula from~\cref{subsec:playing-with}, with $\predHalf = \FHalf{i}(x)$, $\predJ = \Avar{w_i}$ and $\predR$ defined as a union of $\Avar{u_i}$ and $\Avar{v_i}$. Hence, we conclude:
\begin{lemma} \label{lemma:addition}
A well-prepared~$\str{M}$ satisfies~$\phiaddition^{i}(u_i,v_i,w_i)$ iff~$\card{\Avar{u_i}(x)}{\str{M}}{} {+} \card{\Avar{v_i}(x)}{\str{M}}{} {=} \card{\Avar{w_i}(x)}{\str{M}}{}$.
\end{lemma}

The only missing part is the encoding of multiplication. 
Take $\varepsilon_i$ of the form $u_i \cdot v_i = w_i$. 
As already described in the overview, our definition of multiplication requires three steps:
\begin{center}
\begin{description}
\itemsep 0 cm 
\item[(link)] A binary relation~$\Mult{i}{}^{\str{M}}$ links each element from $\Avar{w_i}^{\str{M}}$ to some element from $\Avar{u_i}^{\str{M}}$.
\item[(count)] Each element from $M$ satisfying $\Avar{u_i}(x)$ has exactly $\card{\Avar{v_i}(x)}{\str{M}}{}$ $\Mult{i}{}^{\str{M}}$-successors.
\item[(bfunc)] The binary relation $\Mult{i}{}^{\str{M}}$ is backward-functional.
\end{description}
\end{center}

Such properties can be expressed with the help of $\Maj{}$ quantifier, as presented below:

\begin{center}
\begin{description}
\itemsep 0 cm 
\item[($\philink^i(u_i, w_i)$)] $ \forall{y} \; \Avar{w_i}(y) \rightarrow \exists{x} \; \Mult{i}{(x,y)} \wedge\forall{x} \forall{y} \; \Mult{i}{(x,y)} \rightarrow \left( \Avar{u_i}(x) \wedge \Avar{w_i}(y) \right) $
\item[($\phicount^i(u_i, v_i, w_i)$)] $\forall{x} \; \Avar{u_i}(x) \rightarrow \Maj{y} \left( [\SHalf{i}(y) \wedge \neg \Mult{i}(x,y)] \vee \Avar{v_i}(y) \right)$
\item[($\phibfunc^i(u_i, v_i, w_i)$)] $\forall{x} \; \Avar{w_i}(x) \rightarrow \Maj{y} \left( [\SHalf{i}(y) \wedge x \neq y] \vee \Mult{i}{(y,x)} \right)$
\end{description}
\end{center}
While the first formula, namely $\philink^i(u_i, w_i)$, is immediate to write, the next two are more involved.
A careful reader can notice that they are actually instances of $\phiequality(\predHalf{}, \predR, \predJ)$ formula from~\cref{subsec:playing-with}. In the case of $\phicount^i(u_i, v_i, w_i)$ we have $\predHalf{} = \SHalf{i}$, $\predJ = \Avar{v_i}$ and the $\Mult{i}{}$-successors of $x$ play the role of elements labelled by $\predR$.
For the last formula one can see that we remove exactly one element from $\SHalf{i}$ ($y$ that is equal to $x$) and we replace it with the $\Mult{i}{}$-predecessors of $x$, which implies that there is the unique such predecessor.
We summarise the mentioned facts as follows:

\begin{lemma}\label{lemma:link-count-and-bfunc}
Let $\str{M}$ be a well-prepared structure satisfying $\philink^i(u_i, w_i)$. We have that (i) $\str{M}$ satisfies $\phicount^i(u_i, v_i, w_i)$ iff every $a \in M$ satisfying~$\Avar{u_i}$ is connected via~$\Mult{i}{}$ to exactly~$|\Avar{v_i}|$ elements satisfying~$\Avar{w_i}$ and (ii) $\str{M}$ satisfies $\phibfunc^i(u_i, v_i, w_i) $ iff the binary relation~$\Mult{i}{}^{\str{M}}$ linking elements satisfying~$\Avar{u_i}(x)$ with those satisfying~$\Avar{w_i}(x)$ is backward-functional.
\end{lemma}

Putting the last three properties together, we encode multiplication as their conjunction:
\begin{center}
\begin{description} \label{form:multiplication}
\itemsep 0 cm 
\item[($\phimultiplication^{i}(u_i,v_i,w_i)$)] \;
        $
    \philink^i(u_i, v_i, w_i) \wedge \phicount^i(u_i, v_i, w_i) \wedge \phibfunc^i(u_i, v_i, w_i)
        $
\end{description}
\end{center}
\begin{lemma}\label{lemma:multiplication}
If a well-prepared $\str{M}$ satisfies $\phimultiplication^{i}(u_i,v_i,w_i)$, then~$\card{\Avar{u_i}(x)}{\str{M}}{} {\cdot} \card{\Avar{v_i}(x)}{\str{M}}{} {=} \card{\Avar{w_i}(x)}{\str{M}}{}$.
\end{lemma}

Let~$\phireduction{\varepsilon}$ be~$\phivar$ supplemented with a conjunction of formulae $\phientry$, where $\phientry$ is respectively: 
(i)~$\phiuequalone$ if~$\varepsilon_i$ is equal to~$u_i {=} 1$,
(ii)~$\phihalves^{i} \wedge \phipartition^{i}(u_i,v_i,w_i) \wedge \phiaddition^{i}(u_i,v_i,w_i)$ for~$\varepsilon_i$ of the form~$u_i + v_i = w_i$ and 
(iii)~$\phihalves^{i} \wedge \phipartition^{i}(u_i,v_i,w_i) \wedge \phimultiplication^{i}(u_i,v_i,w_i)$ for~$\varepsilon_i$ of the form~$u_i \cdot v_i = w_i$. 
As the last piece in the proof we show that each solution of the system~$\varepsilon$ corresponds to some model of~$\phireduction{\varepsilon}$. 
Its proof is routine and relies on the correctness of all previously announced facts (consult~\cref{appendix:negative} for more details).
%
Hence, by the undecidability of $\SHTP$, we immediately conclude:
\begin{theorem} \label{thm:fo2undec}
The finite satisfiability problem for $\FOtpercgl$ is undecidable, even when the only percentage quantifier allowed is $\exists^{=50\%}$.
\end{theorem}

Note that in our proof above, all the presented formulas can be easily transformed to formulae under the local semantics of percentage quantifiers as follows.
First, we introduce a fresh binary symbol $U$ and enforce it to be interpreted as the universal relation with $\forall{x} \forall{y} \; U(x,y)$.
Then, we replace every occurrence of $\existsperc{=}{50}x \; \varphi$ by $\existspercwithrel{=}{50}{U}x \; \varphi$.
Obviously, the resulting formula is $\FOt$ formula with local percentage quantifiers.
Thus we conclude:
\begin{corollary}
\label{cor:local-percentage}
The finite satisfiability problem for $\FOtpercloc$ is undecidable.
\end{corollary}

\subsection{Guarded Fragment}
We now focus on the second seminal fragment of $\FO$ considered in this paper, namely on the guarded-fragment~$\GF$.
%
We start from the global semantics of percentage quantifiers.
Consider a unary predicate $\predH$, whose interpretation is constrained to label exactly half of the domain with $\Maj{x}\; \predH(x)$.
We then employ the formula
\[
\forall{x} \; x=x \to \Maj{y} \left[ \predU(x,y) \land \predH(y) \right] \land \Maj{y} \left[ \predU(x,y) \land \neg \predH(y) \right],
\] 
whose satisfaction by $\str{M}$ entails that $U^{\str{M}}$ is the universal relation.
Hence, by putting $\predU$ as a dummy guard in every formula in the undecidability proof of $\FOtpercloc$, we conclude:
\begin{corollary}
The finite satisfiability problem for $\GFpercgl$ is undecidable, even when restricted to its two-variable fragment $\GFtpercgl$.
\end{corollary}

It turns out that the undecidability still holds for $\GF$ once we switch from the global to the local semantics of percentage counting.
In order to show it, we present a reduction from $\GF^3[F]$ (\ie the three-variable fragment of $\GF$ with a distinguished binary $F$ interpreted as a functional relation), whose finite satisfiability was shown to be undecidable in~\cite{Gradel99}.
\begin{theorem}\label{thm:gf-local-undec}
The finite satisfiability problem for $\GFpercloc$ (and even $\GFpercloc^3$) is undecidable.
\end{theorem}
\begin{proof}[Proof sketch.]
By reduction from $\GF^3[F]$ it suffices to express that $F$ is functional.
Let $H, R$ be fresh binary relational symbols.
We use a similar trick to the one from~\cref{subsec:playing-with}, where $H(x, \cdot)$ plays the role of $\predHalf$ (note that $H$ may induce different partitions for different~$x$), $R(\cdot,y)$ plays the role of $\predR$ and $y$ in $x=y$ plays the role of $\predJ$. 

The functionality of $F$ can be expressed with:
\[ \varphi_{\textit{func}} := \forall{x} \; x=x \to [(\forall{y} \;  F(x,y) \to R(x,y)) \land  (\existspercwithrel{=}{50}{R}y \;  H(x,y)) \land \]
\[ \qquad (\forall{y} \; F(x,y) \to (\neg H(x,y) \lor x=y)) \land (\existspercwithrel{=}{50}{R}y \;  ((H(x,y) \land x \neq y) \vee F(x,y)))]  \]
In the appendix we will show that if $\str{M} \models \varphi_{\textit{func}}$ then $F$ is indeed functional and every structure $\str{M}$ with functional $F$ can be extended by $H$ and $R$, such that $\varphi_{\textit{func}}$ holds. \qedhere
\end{proof}

The similar proof techniques do not work for $\GF^2$, since $\GFt$ with counting is decidable~\cite{Pratt-Hartmann07}.
Thus, in the forthcoming section we show that decidability status transfers not only to $\GFt$ with percentage counting, but also with Presburger arithmetics.
This can be then applied to infer decidability of several modal and description logics, see~\cref{app:presburger-ml}.


\section{Positive results}\label{sec:decidable} 

We next show that the finite satisfiability problem for $\GFtpres$ is decidable, as stated below.
\begin{theorem}\label{theo:decidable}
The finite satisfiability problem for $\GFtpres$ is decidable.
\end{theorem}
It is also worth pointing out that \cref{theo:decidable} together with a minor modification of existing techniques~\cite{BaaderBR19} yields decidability of conjunctive query entailment problem for $\GFtpres$, \ie a problem of checking if an existentially quantified conjunction of atoms is entailed by $\GFtpres$ formula. 
This is a fundamental object of study in the area of logic-based knowledge representation. 
All the proofs and appropriate definitions are moved to~\cref{appendix:querying}.
\begin{theorem}\label{thm:result-on-query-answering}
Finite conjunctive query entailment for $\GFtpres$ is decidable.
\end{theorem}

The rest of this section will be devoted to the proof of Theorem~\ref{theo:decidable}, which goes by reduction to the two-variable fragment of $\FO$ with counting quantifiers $\exists^{=k}, \exists^{\leq k}$ for $k \in \N$ with their obvious semantics.
Since the finite satisfiability of $\Ct$ is decidable~\cite{Pratt-Hartmann05},~\cref{theo:decidable} follows.\footnote{Note that we propose a reduction into $\Ct$, not into the \emph{guarded} $\Ct$, which might seem to be more appropriate. As we will see soon, a bit of non-guarded quantification is required in our proof.}

\subsection{Transforming $\GFtpres$ formulae into $\Ct$}
\label{subsec:positive-transformation}

It is convenient to work with formulae in the appropriate normal form.
Following a routine renaming technique (see \eg~\cite{Kazakov04}) 
we can convert in linear time a $\GFtpres$ formula into the following equisatisfiable normal form (over an extended signature):
\begin{eqnarray*}
\label{eq:normal-form-0}
\Psi_0 & := & \forall x \ \gamma(x)
\wedge
\bigwedge_{i=1}^n\Big( \forall x \forall y \ e_i(x,y)  \to  \alpha_i(x,y)\Big)
\wedge
\bigwedge_{i=1}^m \forall x \Big(\sum_{j=1}^{n_i} \lambda_{i,j}\cdot \#_{y}^{r_{i,j}}[x\neq y] \circledast \delta_i\Big),
\end{eqnarray*}
where $\gamma(x)$ and each $\alpha_i(x,y)$ are quantifier-free formulae,
each $e_i(x,y)$ is atomic predicate and all $\lambda_{i,j}$'s and $\delta_i$'s are integers, and $\circledast$ is as in~\cref{subsec:local-presb-quantifiers}.

Then, for every non-null type $\eta$, we replace each of the expressions $\#_{y}^{r_{i,j}}[x\neq y]$ 
with the sum of all the degrees $\deg_{\eta}(x)$ with $\eta$ containing $r_{i,j}(x,y)$,
\ie the sum $\sum_{r_{i,j}(x,y)\in \eta}\ \deg_{\eta}(x)$.
Moreover, since $\bigwedge \forall$ commutes, we obtain the following formula:
\[
\Psi' \ := \ \forall x \ \gamma(x)
\wedge
\bigwedge_{i=1}^n\Big( \forall x \forall y \ e_i(x,y)  \to  \alpha_i(x,y)\Big)
\wedge
\forall x \bigwedge_{i=1}^m \Big(\sum_{j=1}^{n_i} \lambda_{i,j}\cdot\sum_{r_{i,j}(x,y)\in \eta} \deg_{\eta}(x) \ \circledast \ \delta_i\Big)
\]

Note that the conjunction $\bigwedge_{i=1}^m \Big(\sum_{j=1}^{n_i} \lambda_{i,j}\cdot\sum_{r_{i,j}(x,y)\in \eta} \deg_{\eta}(x) \ \circledast \ \delta_i\Big)$
is a Presburger formula with free variables $\deg_{\eta}(x)$'s, for every non-null type $\eta$.\footnote{Technically speaking, 
in the standard definition of Presburger formula,
the equality $f\equiv_d g$ is not allowed. 
However, it can be rewritten as $\exists x_1 \exists x_2 ( f +x_1d = g+x_2d)$.}
Thus, by Theorem~\ref{theo:semilinear}, we can compute 
a set of tuples of vectors $\{(\vv_{1,0},\vc_{1,1},\ldots,\vv_{1,k_1}),\ldots,(\vv_{p,0},\vv_{p,1},\ldots,\vv_{p,k_p})\}$
and further rewrite $\Psi'$ into the following formula:
\[
\Psi =
\forall x \ \gamma(x)
\ \wedge\
\textstyle\bigwedge_{i=1}^n\Big( \forall x \forall y \ e_i(x,y)  \to  \alpha_i(x,y)\Big)
\ \wedge\ 
\forall x \ \deg(x) \in S
\]
where $S = \bigcup_{i=1}^{p}\ L(\vc_{i,0};\vc_{i,1},\ldots,\vc_{i,k_i})$.
We stress that technically $\Psi$ is no longer in $\GFtpres$.

In the following we will show how to transform $\Psi$ into a $\Ct$ formula $\Psi^*$
such that they are (finitely) equi-satisfiable.
For every $i=1,\ldots,p$, 
let $S_i = L(\vv_{i,0};\vv_{i,1},\ldots,\vv_{i,k_i})$.
Recall that $\offset {S_i}$ is the offset vector $\vv_{i,0}$
and $\prd {S_i}$ is the set of periodic vectors of $S_i$, \ie $\{\vv_{i,1},\ldots,\vv_{i,k_i}\}$.
Consider the following formulae $\xi$ and $\phi$.
\[
\xi \; := \;
\forall x 
\bigvee_{i=1}^p
\underline{\deg(x) {=} \offset {S_i}}
\ \vee \
\underline{\deg(x) {\in} \prd {S_i}}, \; \;
\phi \; := \;
\forall x \bigwedge_{i=1}^p\underline{\deg(x) {\neq} \offset {S_i}}
\ {\to} \ \exists y \ \varphi(x,y)
\]
where $\varphi(x,y)$ is the conjunction expressing the following properties:
\begin{itemize}
\item 
The $1$-types of $x$ and $y$ equal. It can be expressed with the formula $\bigwedge_{U} U(x)\leftrightarrow U(y)$, where $U$ ranges over unary predicates appearing in $\Psi$.
\item
$\deg(x)\in \prd{S_j}$ and $\deg(y)=\offset {S_j}$
for some $1\leq j \leq p$.
\end{itemize}
Note that $\underline{\deg(x)= \offset {S_i}}$ can be written as a $\Ct$ formula.
For example, if $\vv_{i,0}=(d_1,\ldots,d_{\ell})$,
it is written as $\bigwedge_{j=1}^{\ell} \exists^{=d_{j}}y \ \eta_j(x,y)$.
We can proceed with $\underline{\deg(x)\in \prd {S_i}}$ similarly, since $\prd{S_i}$ contains only finitely many vectors. 
Finally, we put $\Psi^*$ to be
\begin{eqnarray*}
\label{eq:c2-reduction}
\Psi^* & := & \forall x \ \gamma(x)
\ \wedge\
\bigwedge_{i=1}^n \forall x \forall y \ e_i(x,y)  \to  \alpha_i(x,y)
\quad
\wedge\quad\xi\quad\wedge\quad\phi.
\end{eqnarray*}

We will show that $\Psi$ and $\Psi^*$ are finitely equi-satisfiable, as stated formally below.
\begin{lemma}
\label{lem:fin-equi-sat}
$\Psi$ is finitely satisfiable if and only if
$\Psi^*$ is.
\end{lemma}

We delegate the proof of~\cref{lem:fin-equi-sat} to the next section.
We conclude by stating that the complexity of our decision procedure is $\ThreeNExpTime$.
For more details of our analysis, see~\cref{subsec:complexity}.
Note that if we follow the decision procedure described in~\cite{semilinear}
for converting a system of linear equations to its semilinear set representation
we will obtain a non-elementary complexity.
This is because we need to perform $k{-}1$ intersections,
where $k$ is the number of linear constraints in the formula $\Psi'$,
and the procedure in~\cite{semilinear} for handling each intersection 
yields an exponential blow-up.
Instead, we use the results in~\cite{Domenjoud91,Pottier91,ChistikovH16}
and obtain the complexity $\ThreeNExpTime$, which though still high, 
falls within the elementary class.

\subsection{Correctness of the translation}

Before we proceed with the proof,
we need to define some terminology.
Let $\str{M}$ be a finite model.
Let $a,b\in A$ be such that the 2-type of $(a,b)$ is $\etanull$, \ie the null-type
and that $a$ and~$b$ have the same 1-type.
Suppose $c_1,\ldots,c_s$ are all elements such that the 2-type of each $(a,c_j)$,
denoted by $\eta_j'$, is non-null.
Likewise, $d_1,\ldots,d_t$ are all the elements such that 
the 2-type of $(b,d_j)$, denoted by $\eta_j''$, is non-null.
Moreover, $c_1,\ldots,c_s,d_1,\ldots,d_t$ are pair-wise different.

``Merging'' $a$ and $b$ into one new element $\hat{a}$
is defined similarly to the one in the graph-theoretic sense
where $a$ and $b$ are merged into $\hat{a}$ such that the following holds.
\begin{itemize}
\item 
The 2-types of each $(\hat{a},c_j)$ are equal to the original 2-types of $(a,c_j)$,
for all~$j=1,\ldots,s$.
\item 
The 2-types of each $(\hat{a},d_j)$ are equal to the original 2-types of $(a,d_j)$,
for all~$j=1,\ldots,t$.
\item
The 2-types of $(\hat{a},a')$ is the null type, for every $a'\notin \{c_1,\ldots,c_2,d_1,\ldots,d_t\}$. 
\item
The 1-type of $\hat{a}$ is  the original 1-type of $a$
(which is the same as the 1-type of $b$).
\end{itemize}

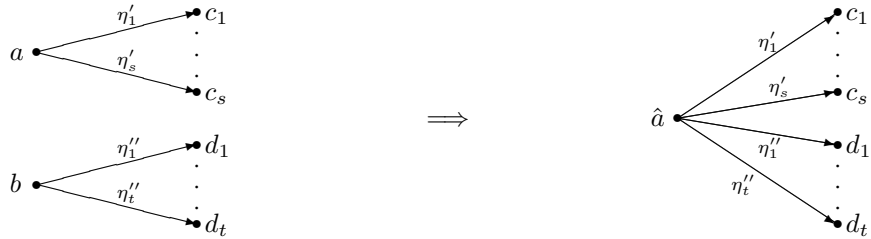
\begin{center}
\begin{picture}(340,75)(0,15)


\put(20,75){\circle*{3}}\put(10,72){$a$}
\put(20,75){\vector(4,1){60}}\put(50,87){\scriptsize $\eta_1'$}
\put(80,90){\circle*{3}}\put(83,87){$c_1$}
\multiput(80,82)(0,-8){3}{\circle*{1}}
\put(20,75){\vector(4,-1){60}}\put(50,70){\scriptsize $\eta_s'$}
\put(80,60){\circle*{3}}\put(83,57){$c_s$}

\put(20,25){\circle*{3}}\put(10,22){$b$}
\put(20,25){\vector(4,1){60}}\put(50,37){\scriptsize $\eta_1''$}
\put(80,40){\circle*{3}}\put(83,37){$d_1$}
\multiput(80,32)(0,-8){3}{\circle*{1}}
\put(20,25){\vector(4,-1){60}}\put(50,20){\scriptsize $\eta_t''$}
\put(80,10){\circle*{3}}\put(83,7){$d_t$}

\put(166,47){$\Longrightarrow$}

\put(260,50){\circle*{3}}\put(250,47){$\hat{a}$}

\qbezier(260,50)(290,70)(320,90)\put(319,89){\vector(3,2){0}}
\put(290,77){\scriptsize $\eta_1'$}
\put(320,90){\circle*{3}}\put(323,87){$c_1$}

\multiput(320,82)(0,-8){3}{\circle*{1}}

\qbezier(260,50)(290,55)(320,60)\put(319,60){\vector(4,1){0}}
\put(294,60){\scriptsize $\eta_s'$}
\put(320,60){\circle*{3}}\put(323,57){$c_s$}

\qbezier(260,50)(290,45)(320,40)\put(319,40){\vector(4,-1){0}}
\put(290,37){\scriptsize $\eta_1''$}
\put(320,40){\circle*{3}}\put(323,37){$d_1$}

\multiput(320,32)(0,-8){3}{\circle*{1}}

\qbezier(260,50)(290,30)(320,10)\put(319,10){\vector(3,-2){0}}
\put(280,23){\scriptsize $\eta_t''$}
\put(320,10){\circle*{3}}\put(323,7){$d_t$}

\end{picture}
\end{center}
Note that we require that the original 2-type of $(a,b)$ is the null type.
Thus, after the merging, the degree of $\hat{a}$
is the sum of the original degrees of $a$ and $b$.
Moreover, the 1-type of $\hat{a}$ is the same as the original 1-type of $a$ and $b$.
Thus, if $\forall x \forall y\ e_i(x,y)\to \alpha_i(x,y)$ holds in $\str{M}$,
after the merging, it will still hold.
Likewise, if $\forall x\ \gamma(x)$ holds in $\str{M}$,
it will still hold after the merging.

For the inverse,
we define the ``splitting'' of an element $\hat{a}$ into two elements $a$ and $b$
as illustrated above,
where the 1-type of $a$ and $b$ is the same as the 1-type of $\hat{a}$
and the 2-type of $(a,b)$ is set to be $\etanull$.
After the splitting,
the sum of the degrees of $a$ and $b$
is the same as the original degree of $\hat{a}$.
Moreover, since the 2-type of $(a,b)$ is $\etanull$,
$\str{M},x/a,y/b\not\models e_i(x,y)$.
Thus, if $\forall x \forall y\ e_i(x,y)\to \alpha_i(x,y)$ holds in the original $\str{M}$, it will still hold after the splitting.

\begin{lemma}\label{lem:fin-equi-sat-only-if}
If $\Psi$ is finitely satisfiable then $\Psi^*$ is.
\end{lemma}
\begin{proof}
Let $\str{M}$ be a finite model of $\Psi$.
We will construct a finite model $\str{M}^*\models \Psi^*$
by splitting every element in $\str{M}$ into several elements so that
their degrees are either one of the offset vectors of $S$
or one of the period vectors.

Let $a\in A$ and $\deg_{\str{M}}(a)\in S_i$, for some $1\leq i \leq p$.
Suppose $\deg_{\str{M}}(a) = \vv_{i,0} + \sum_{j=1}^{k_i} n_j \vv_{i,j}$,
for some $n_1,\ldots,n_{k_i}\geq 0$.
Let $N=1+\sum_{j=1}^{k_i} n_j$.
We split $a$ into $N$ elements $b_1,\ldots,b_N$.
Let $\str{M}^*$ denote the resulting model after such splitting.
Note that it should be finite since the degree of $a$ is finite.
It is straightforward to show that $\str{M}^*\models \Psi^*$.
\end{proof}

\begin{lemma}\label{lem:fin-equi-sat-if}
If $\Psi^*$ is finitely satisfiable then $\Psi$ is.
\end{lemma}
\begin{proof}
Let $\str{M}^*$ be a finite model of $\Psi^*$.
Note that the degree of every element in $\str{M}^*$ is
either the offset vector or one of the period vectors of $S_i$,
for some $1\leq i \leq p$.
To construct a finite model $\str{M}\models \Psi$,
we can appropriately ``merge'' elements so
that the degree of every element is a vector in $S_i$,
for some $1\leq i \leq p$.

To this end, we call an element $a$ in $\str{M}^*$
a {\em periodic} element, if its degree is not an offset vector of some $S_i$.
Let $N$ be the number of periodic elements in $\str{M}^*$.
We make $3N$ copies of $\str{M}^*$, which we denote by $\str{M}_{i,j}$,
where $0\leq i \leq 2$ and $1\leq j \leq N$.
Let $\str{M}$ be a model obtained by the disjoint union of all of $\str{M}_{i,j}$'s,
where for every $b,b'$ that do not come from the same $\str{M}_{i,j}$,
the 2-type of $(b,b')$ is the null-type.

We will show how to eliminate periodic elements in $\str{M}$
by appropriately ``merging'' its elements.
We need the following terminology.
Recall that $S=S_1\cup \cdots \cup S_p$, where each $S_i$ is a linear set.
For two vectors $\vu$ and $\vv$,
we say that {\em $\vu$ and $\vv$ are compatible} (w.r.t. the semilinear set $S$),
if there is $S_i$ such that $\vu$ is the offset vector of $S_i$
and $\vv$ is one of the period vectors of $S_i$.
We say that two elements $a$ and $b$ in $\str{M}$ are {\em merge-able},
if their $1$-types are the same and their degrees are compatible.

\newcommand{\myoval}[6]{
\qbezier(#1,#5)(#1,#6)(#2,#6)
\qbezier(#2,#6)(#3,#6)(#3,#5)
\qbezier(#3,#5)(#3,#4)(#2,#4)
\qbezier(#2,#4)(#1,#4)(#1,#5)
}

We show how to merge  periodic elements in $\str{M}_{0,j}$, for every $j=1,\ldots,N$.
\begin{itemize}
\item
Let $b_1,\ldots,b_N$ be the periodic elements in $\str{M}_{0,j}$.
\item
For each $l=1,\ldots,N$, let $a_l$ be an offset element in $\str{M}_{1,l}$
such that every $b_l$ and $a_l$ are merge-able.
(Such $b_l$ exists, since $\str{M}^*$ satisfies $\Psi^*$ and each $\str{M}_{i,j}$ is isomorphic to $\str{M}^*$.)
\item
Then, merge $a_l$ and $b_l$ into one element, for every $l=1,\ldots,k$.
\end{itemize}
See below, for an illustration for the case when $j=1$.
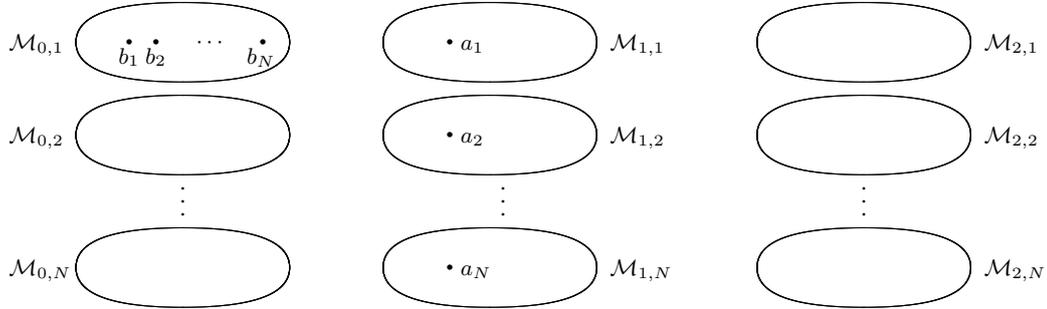
\begin{center}
\begin{picture}(400,95)(0,85)
{\footnotesize 
\put(5,177){$\str{M}_{0,1}$}
\myoval{30}{70}{110}{165}{180}{195}
\put(50,180){\circle*{2}}\put(46,172){\footnotesize $b_1$}
\put(60,180){\circle*{2}}\put(56,172){\footnotesize $b_2$}
\put(75,178){$\cdots$}
\put(100,180){\circle*{2}}\put(94,172){\footnotesize $b_N$}

\put(5,142){$\str{M}_{0,2}$}
\myoval{30}{70}{110}{130}{145}{160}

\multiput(70,115)(0,5){3}{\circle*{1}}

\put(5,92){$\str{M}_{0,N}$}
\myoval{30}{70}{110}{80}{95}{110}


\put(230,177){$\str{M}_{1,1}$}
\myoval{145}{185}{225}{165}{180}{195}
\put(170,180){\circle*{2}}\put(174,177){\footnotesize $a_1$}

\put(230,142){$\str{M}_{1,2}$}
\myoval{145}{185}{225}{130}{145}{160}
\put(170,145){\circle*{2}}\put(174,142){\footnotesize $a_2$}

\multiput(190,115)(0,5){3}{\circle*{1}}

\put(230,92){$\str{M}_{1,N}$}
\myoval{145}{185}{225}{80}{95}{110}
\put(170,95){\circle*{2}}\put(174,92){\footnotesize $a_N$}


\put(370,177){$\str{M}_{2,1}$}
\myoval{285}{325}{365}{165}{180}{195}

\put(370,142){$\str{M}_{2,2}$}
\myoval{285}{325}{365}{130}{145}{160}

\multiput(325,115)(0,5){3}{\circle*{1}}

\put(370,92){$\str{M}_{2,N}$}
\myoval{285}{325}{365}{80}{95}{110}
}
\end{picture}
\end{center}

Obviously, after this merging, there is no more periodic element in $\str{M}_{0,j}$, for every $j=1,\ldots,N$.
We can perform similar merging between the periodic elements in 
$\str{M}_{1,1}\cup\cdots \cup \str{M}_{1,N}$ and the offset elements in $\str{M}_{2,1}\cup\cdots \cup \str{M}_{2,N}$,
and between the periodic elements in 
$\str{M}_{2,1}\cup\cdots \cup \str{M}_{2,N}$ and the offset elements in $\str{M}_{0,1}\cup\cdots \cup \str{M}_{0,N}$.

After such merging, there is no more periodic element in $\str{M}$
and the degree of every element is now a vector in $S_i$, for some $1\leq i \leq p$.
Moreover, since the merging preserves the satisfiability of $\forall x \ \gamma(x)$
and each $\forall x \forall y \ e_i(x,y)  \to  \alpha_i(x,y)$,
the formula $\Psi$ holds in $\str{M}$.
That is, $\Psi$ is finitely satisfiable.
\end{proof}

%
%
%
%
%
%
%
%
%

\subsection{Complexity analysis of the decision procedure}
\label{subsec:complexity}

We need to introduce more terminology.
For a vector/matrix $X$, we write $\|X\|$ to denote its $L_{\infty}$-norm,
\ie the maximal absolute value of its entries.
For a set of vector/matrices $B$,
we write $\|B\|$ to denote $\max_{X\in B} \|X\|$. 

Let $P =\{\vv_1,\ldots,\vv_k\} \subseteq \N^{\ell}$ 
be a finite set of (row) vectors of natural number components.
To avoid clutter, we write $L(\vu;P)$ to denote the linear set $L(\vu;\vv_1,\ldots,\vv_k)$.
For a finite set $B\subseteq \N^{\ell}$, we write
$L(B;P)$ to denote the set $\bigcup_{\vu\in B} L(\vu;P)$.

We will use the following fact from~\cite{Domenjoud91,Pottier91}.
See also Proposition~2 in~\cite{ChistikovH16}.

\begin{proposition}
\label{prop:semi-linear}
Let $A\in \Z^{\ell\times m}$ and $\vc\in\Z^{m}$.
Let $\Gamma$ be the space of the solutions of the system $\vx A =\vc$ 
(over the set of natural numbers $\N$).\footnote{Recall that vectors in this paper are row vectors. 
So, $\vx$ and $\vc$ are row vectors of $\ell$ variables and $m$ constants, respectively.}
Then, there are finite sets $B,P \subseteq \N^{\ell}$ such that the following holds.
\begin{itemize}
\item
$L(B;P)=\Gamma$.
\item
$\|B\| \leq ((m+1)\|A\|+\|\vc\|+1)^{\ell}$.
\item
$\|P\| \leq (m\|A\|+1)^{\ell}$.
\item
$|B|\leq (m+1)^{\ell}$.
\item
$|P|\leq m^{\ell}$.
\end{itemize}
\end{proposition}
By repeating some of the vectors, if necessary,
we can assume that Proposition~\ref{prop:semi-linear}
states that $|B|=|P|=(m+1)^{\ell}$.

Proposition~\ref{prop:semi-linear} immediately implies
the following na\"ive construction of the sets $B$ and $P$ in deterministic double-exponential time
(in the size of input $A$ and $\vc$).
\begin{itemize}
\item
Enumerate all possible sets $B,P\subseteq \N^{\ell}$ of cardinality $(m+1)^{\ell}$
whose entries are all bounded above by $((m+1)\|A\|+\|\vc\|+1)^{\ell}$.
\item
For each pair $B,P$, where $P=\{\vv_1,\ldots,\vv_k\}$,
check whether for every $i_1,\ldots,i_k\in\N$ and every $\vu \in B$,
the following equation holds.
\begin{eqnarray}
\label{eq:sanity}
(\vu + \sum_{j=1}^k i_j\vv_j)A & = & \vc.
\end{eqnarray}
\end{itemize}
The number of bits needed to represent the sets $B$ and $P$ is
$O(\ell^2(m+1)^{\ell}\log K)$, where $K=(m+1)\|A\|+\|\vc\|+1$.
Since Eq.~\ref{eq:sanity} can be checked in deterministic exponential time 
(more precisely, it takes non-deterministic polynomial time to check if there is $i_1,\ldots,i_k$ such that Eq.~\ref{eq:sanity} does not hold)
 in the length of the bit representation 
of the vectors in $B$, $P$, $A$ and the vector $\vc$,
see, \eg~\cite{Papadimitriou81},
constructing the sets $B$ and $P$ takes double-exponential time.

For completeness, we repeat the complexity analysis in Section~\ref{sec:decidable}.
First, the formula $\Psi_0$ takes linear time in the size of the input formula.
Constructing the formula $\Psi'$ requires exponential time (in the number of binary predicates),
\ie $\ell = 2^k-1$, where $k$ is the number of binary predicates.
Thus, constructing the sets $B$ and $P$ takes deterministic triple exponential time in the size of $\Psi_0$.
However, the size of $B$ and $P$ is $O(2^{2k}(m+1)^{2^k}\log K)$, \ie
double exponential in the size of $\Psi_0$.
The $\Ct$ formulas $\xi$ and $\phi$ are constructed in polynomial time in the size of $B$ and $P$.
Since both the satisfiability and finite satisfiability of $\Ct$ formulas
is decidable in nondeterministic exponential time, we have another exponential blow-up.
Altogether, our decision procedure runs in $\ThreeNExpTime$.

\section{Concluding remarks}

In the paper we studied the finite satisfiability problem for classical decidable fragments of $\FO$ extended with percentage quantifiers (as well as arithmetics in the full generality), namely the two-variable fragment $\FOt$ and the guarded fragment $\GF$. 
We have shown that even in the presence of percentage quantifiers they quickly become undecidable.

The notable exception is the intersection of $\GF$ and $\FOt$, \ie the two-variable guarded fragment, for which we have shown that it is decidable with elementary complexity, even when extended with local Presburger arithmetics. 
The proof is quite simple and goes via an encoding into the two-variable logic with counting ($\Ct$).
One of the bottlenecks in our decision procedure is the conversion of 
systems of linear equations into the semilinear set representations,
which incurs a double-exponential blow-up.
We leave it for future work whether a decision procedure with lower complexity is possible
and/or whether the conversion to semilinear sets is necessary.

\noindent We stress that our results are also applicable to the unrestricted satisfiability problem (whenever the semantics of percentage quantifiers make sense), see~\cref{app:gen-sat}.

\bibliography{references}

\begin{thebibliography}{10}

\bibitem{AndrekaNB98}
Hajnal Andr{\'{e}}ka, Istv{\'{a}}n N{\'{e}}meti, and Johan van Benthem.
\newblock {M}odal {L}anguages and {B}ounded {F}ragments of {P}redicate {L}ogic.
\newblock {\em J. Philosophical Logic}, 1998.

\bibitem{Baader17}
Franz Baader.
\newblock A new description logic with set constraints and cardinality
  constraints on role successors.
\newblock In Clare Dixon and Marcelo Finger, editors, {\em {FroCoS}}, 2017.

\bibitem{BaaderBR19}
Franz Baader, Bartosz Bednarczyk, and Sebastian Rudolph.
\newblock Satisfiability checking and conjunctive query answering in
  description logics with global and local cardinality constraints.
\newblock In {\em {DL}}, 2019.

\bibitem{BaaderBR20}
Franz Baader, Bartosz Bednarczyk, and Sebastian Rudolph.
\newblock Satisfiability and query answering in description logics with global
  and local cardinality constraints.
\newblock In {\em {ECAI}}, 2020.

\bibitem{DLBook}
Franz Baader, Ian Horrocks, Carsten Lutz, and Ulrike Sattler.
\newblock {\em An Introduction to Description Logic}.
\newblock Cambridge University Press, 2017.

\bibitem{BaranyCS15}
Vince B{\'{a}}r{\'{a}}ny, Balder ten Cate, and Luc Segoufin.
\newblock Guarded negation.
\newblock {\em J. {ACM}}, 2015.

\bibitem{ArxivVersion}
Bartosz Bednarczyk, Maja Orłowska, Anna Pacanowska, and Tony Tan.
\newblock {On Classical Decidable Logics extended with Percentage Quantifiers
  and Arithmetics}.
\newblock {\em CoRR}, abs/2106.15250, 2021.
\newblock URL: \url{https://arxiv.org/abs/2106.15250}, \href
  {http://arxiv.org/abs/2106.15250} {\path{arXiv:2106.15250}}.

\bibitem{BenediktKT20}
Michael Benedikt, Egor~V. Kostylev, and Tony Tan.
\newblock Two variable logic with ultimately periodic counting.
\newblock In {\em {ICALP} 2020}, 2020.

\bibitem{BorgerGG1997}
Egon B{\"{o}}rger, Erich Gr{\"{a}}del, and Yuri Gurevich.
\newblock {\em The Classical Decision Problem}.
\newblock Perspectives in Mathematical Logic. Springer, 1997.

\bibitem{ChistikovH16}
Dmitry Chistikov and Christoph Haase.
\newblock The taming of the semi-linear set.
\newblock In {\em {ICALP}}, 2016.

\bibitem{DemriL10}
St{\'{e}}phane Demri and Denis Lugiez.
\newblock Complexity of modal logics with presburger constraints.
\newblock {\em J. Appl. Log.}, 2010.

\bibitem{Domenjoud91}
Eric Domenjoud.
\newblock Solving systems of linear diophantine equations: An algebraic
  approach.
\newblock In {\em {MFCS}}, 1991.

\bibitem{semilinear}
Seymour Ginsburg and Edwin~Henry Spanier.
\newblock Semigroups, presburger formulas, and languages.
\newblock {\em Pac. J. of Math.}, 16:285--296, 1966.

\bibitem{GlimmLHS08}
Birte Glimm, Carsten Lutz, Ian Horrocks, and Ulrike Sattler.
\newblock {C}onjunctive {Q}uery {A}nswering for the {D}escription {L}ogic
  {SHIQ}.
\newblock {\em J. Artif. Intell. Res. 2008}, 2008.

\bibitem{Gradel:DL}
Erich Gr{\"{a}}del.
\newblock Description logics and guarded fragments of first order logic.
\newblock In {\em {DL}}, 1998.

\bibitem{Gradel99}
Erich Gr{\"{a}}del.
\newblock On the restraining power of guards.
\newblock {\em J. Symb. Log.}, 1999.

\bibitem{GradelKV97}
Erich Gr{\"{a}}del, Phokion~G. Kolaitis, and Moshe~Y. Vardi.
\newblock On the decision problem for two-variable first-order logic.
\newblock {\em Bulletin of Symbolic Logic}, 1997.

\bibitem{GradelOR97}
Erich Gr{\"{a}}del, Martin Otto, and Eric Rosen.
\newblock Two-variable logic with counting is decidable.
\newblock In {\em {LICS}}, 1997.

\bibitem{Gradel:undec}
Erich Gr{\"{a}}del, Martin Otto, and Eric Rosen.
\newblock Undecidability results on two-variable logics.
\newblock {\em Arch. Math. Log.}, 1999.

\bibitem{Kazakov04}
Yevgeny Kazakov.
\newblock A polynomial translation from the two-variable guarded fragment with
  number restrictions to the guarded fragment.
\newblock In {\em {JELIA}}, volume 3229 of {\em LNCS}, 2004.

\bibitem{KupkeP10}
Clemens Kupke and Dirk Pattinson.
\newblock On modal logics of linear inequalities.
\newblock In Lev~D. Beklemishev, Valentin Goranko, and Valentin~B. Shehtman,
  editors, {\em {AIML}}, 2010.

\bibitem{Libkin04}
Leonid Libkin.
\newblock {\em Elements of Finite Model Theory}.
\newblock Texts in Theoretical Computer Science. An {EATCS} Series. Springer,
  2004.

\bibitem{Matiyasevich}
Yuri~V. Matiyasevich.
\newblock {\em Hilbert's Tenth Problem}.
\newblock MIT Press, 1993.

\bibitem{PracaLic}
Anna Pacanowska and Maja Orlowska.
\newblock {Statystyczne konstrukcje w rozstrzygalnych fragmentach logiki
  pierwszego rzędu}.
\newblock Bachelor's thesis, University of Wrocław, 2021.

\bibitem{PacholskiST00}
Leszek Pacholski, Wieslaw Szwast, and Lidia Tendera.
\newblock Complexity results for first-order two-variable logic with counting.
\newblock {\em {SIAM} J. Comput.}, 29(4):1083--1117, 2000.

\bibitem{Papadimitriou81}
Christos~H. Papadimitriou.
\newblock On the complexity of integer programming.
\newblock {\em J. {ACM}}, 28(4):765--768, 1981.

\bibitem{Pottier91}
Loic Pottier.
\newblock Minimal solutions of linear diophantine systems: Bounds and
  algorithms.
\newblock In {\em {RTA}}, 1991.

\bibitem{Pratt-Hartmann05}
Ian Pratt{-}Hartmann.
\newblock Complexity of the two-variable fragment with counting quantifiers.
\newblock {\em J. Log. Lang. Inf.}, 14(3):369--395, 2005.

\bibitem{Pratt-Hartmann07}
Ian Pratt{-}Hartmann.
\newblock Complexity of the guarded two-variable fragment with counting
  quantifiers.
\newblock {\em J. Log. Comput.}, 17(1):133--155, 2007.

\bibitem{Scott62}
Dana Scott.
\newblock A decision method for validity of sentences in two variables.
\newblock {\em Journal of Symbolic Logic}, 1962.

\bibitem{Tessaris01}
Sergio Tessaris.
\newblock {Querying expressive DLs}.
\newblock In {\em {DL} 2001}, 2001.

\bibitem{trakhtenbrot}
B.~Trakhtenbrot.
\newblock The impossibility of an algorithm for the decidability problem on
  finite classes.
\newblock In {\em Proc. USSR Acad. Sci.}, volume~70, pages 569--572.

\bibitem{Quine76}
Willard van Orman~Quine.
\newblock {\em The Ways of Paradox and Other Essays, Revised Edition}.
\newblock Harvard University Press, 1976.

\end{thebibliography}
\clearpage
\appendix

\section{From Presburger modal logic with converse to $\GFt$ with Presburger}\label{app:presburger-ml}

In this section we give a rather standard translation from the (multi)modal logic with converse, extended by Presburger arithmetics~\cite{DemriL10} into $\GFtpres$. A similar translation can be defined also for $\ALCSCC$ with inverse from~\cite{Baader17}, whose semantics is a bit more complicated but essentially means the same.

\newcommand{\prop}{\textsf{Prop}}
\newcommand{\rel}{\textsf{Rel}}

We employ a countable set of propositional variables $\prop = \{ p_1, p_2, \ldots \}$ and a countable 
set of relation symbols $\rel = \{ R_1, R_2, \ldots \}$. 
We define the formulae of Presburger Modal Logic with converse, as follows:
\begin{align*}
    \varphi ::= & \; p  \; \mid \; \neg \varphi \; \mid \; \varphi \land \varphi \; \mid \; t \sim b \; \mid \; t \equiv_k c \\
    t ::= & \; a \cdot \#^R \varphi \; \mid \; a \cdot \#^{R^-} \varphi \; \mid \; t + t    
\end{align*}
where $p \in \prop$, $R \in \rel$, $b, c \in \mathbb{N}$, $k \in \mathbb{N} \setminus \{ 0,1 \}$, $a \in \mathbb{Z} \setminus \{ 0 \}$, $\sim \;\in \{ \leq, <, >, =, \geq \}$.

A Kripke structure is a tuple $\mathcal{M} = (W, (R^{\str{M}})_{R \in \rel}, \ell)$, where $W$ is a set of \emph{worlds}, each $R^{\mathcal{M}}$ is a binary relation and $\ell : W \to 2^{\prop}$ label worlds with atomic propositions.
Now, for the semantics, we define $\mathcal{M}, w \models \varphi$ as follows:
\begin{itemize}
    \item $\mathcal{M}, w \models p$ iff $p \in \ell(w)$ 
    \item $\mathcal{M}, w \models \neg \varphi$ iff not $\mathcal{M}, w \models \varphi$
    \item $\mathcal{M}, w \models \varphi \land \varphi'$ iff $\mathcal{M}, w \models \varphi$ and $\mathcal{M}, w \models \varphi'$
    \item $\mathcal{M}, w \models \Sigma_{i} a_i \cdot \#^{S_i} \varphi_i \;\sim\; b $ iff 
    $\left( \Sigma_{i} a_i \cdot | \{ w' \mid (w,w') \in S_i^{\str{M}} \} | \right) \; \sim \; b$ (where in the case of $S_i = R^-$ we think about the inverse relation of $R^{\str{M}}$). And similarly for $\equiv_k c$ (meaning the congruence modulo $k$).
\end{itemize}
We say that $\varphi$ is globally satisfiable if there is a structure structure $\str{M}$ such that for all its worlds $w$ we have $\str{M}, w \models \varphi$.
 
\newcommand{\tr}{\mathfrak{tr}}
We next define a translation of Presburger modal logic into $\GFtpres$. We let $v$ denote either $x$ or $y$ and $\bar{v}$ the other variable.
\begin{itemize}
    \item $\tr_v(p) = p(v)$
    \item $\tr_v(\neg \varphi) = \neg \tr_v(\varphi)$
    \item $\tr_v(\varphi \land \varphi') = \tr_v(\varphi) \land \tr_v(\varphi')$
    \item $\tr_v(\Sigma_{i} a_i \cdot \#^{S_i} \varphi_i \;\sim\; b) = \sum_{i=1}^n\ a_i\cdot \#_{\bar{v}}^{S_i}[\tr_{\bar{v}}(\varphi)] \sim b$
    \item $\tr_v(\Sigma_{i} a_i \cdot \#^{S_i} \varphi_i \;\equiv_k\; c) = \sum_{i=1}^n\ a_i\cdot \#_{\bar{v}}^{S_i}[\tr_{\bar{v}}(\varphi)] \equiv_k c$
\end{itemize}
Hence, for a given $\varphi$ let $\tr(\varphi) = \forall{x} \; x=x \to \tr_x(\varphi)$.
It is easy to see that $\tr(\varphi)$ has a (finite) model iff $\varphi$ is globally (finitely) satisfiable.
Hence, we obtain a logspace reduction from the global satisfiability for Presburger modal logic with converse into $\GFtpres$.

\section{Appendix for~\cref{sec:negative_results}}\label{appendix:negative}

\subsection{Proof of~\cref{fact:equality}}

\begin{proof}
The fact that a finite $(\predHalf, \predR, \predJ)$-separated $\str{A}$ such that $\card{\predR(x)}{\str{A}}{} = \card{\predJ(x)}{\str{A}}{}$ satisfies $\phiequality(\predHalf, \predR, \predJ)$ is obvious from the semantics of $\FOtpercgl$ and the definition of $(\predHalf, \predR, \predJ)$-separability. 
For the opposite direction let $\str{A}$ be a finite, $(\predHalf, \predR, \predJ)$-separated model of $\phiequality(\predHalf, \predR, \predJ)$. Then we can see that, by the satisfaction of $\Maj$, $\cA$ has an even number of elements (call it $2n$). Moreover, we know that the sets $\predR^{\cA}$ and $\predJ^{\cA}$ are disjoint, hence by the satisfaction of $\Maj{x}\; (\predHalf(x) \wedge \neg \predR(x)) \vee \predJ(x)$ we conclude that $n = (n - \card{\predR(x)}{\str{A}}{}) + \card{\predJ(x)}{\str{A}}{}$, which implies that $\card{\predR(x)}{\str{A}}{} = \card{\predJ(x)}{\str{A}}{}$.
\end{proof}

\subsection{Proof of~\cref{fact:phiuequalone}, \cref{fact:phihalves} and~\cref{fact:phipartition}}
\begin{proof} 
An immediate consequence of the semantics of $\FOtpercgl{}$. 
\end{proof}

\subsection{Proof of~\cref{lemma:addition}}
\begin{proof}
Identical to the proof~\cref{fact:equality}, by taking $\predHalf = \FHalf{i}(x)$, $\predJ = \Avar{w_i}$ and $\predR$ defined as the union of $\Avar{u_i}$ and $\Avar{v_i}$.
\end{proof}

\subsection{Proof of~\cref{lemma:link-count-and-bfunc}}
\begin{proof}
Observe $\phicount^i(u_i, v_i, w_i)$ that is actually an instance of $\phiequality(\predHalf{}, \predR, \predJ)$ formula from~\cref{subsec:playing-with} with $\predHalf{} = \SHalf{i}$, $\predJ = \Avar{v_i}$ and the $\Mult{i}{}$-successors of $x$ play the role of elements labelled by $\predR$.
For $\phibfunc^i(u_i, v_i, w_i)$ we proceed similarly.
\end{proof}

\subsection{Proof of~\cref{lemma:multiplication}}

Indeed, assume that a well prepared $\str{M}$ satisfies $\phimultiplication^{i}(u_i,v_i,w_i)$.
Then from~\cref{lemma:link-count-and-bfunc}(i) it follows that every domain element labelled by $\Avar{u_i}$ is connected via $\Mult{i}{}^{\str{M}}$ relation with exactly $|\Avar{v_i}^{\str{M}}|$ that satisfies $\Avar{w_i}$.
By backward-functionality of $\Mult{i}{}^{\str{M}}$ (guaranteed by~\cref{lemma:link-count-and-bfunc}(ii)) the sets of elements labelled with 
$\Avar{w_i}$ and connected to elements satisfying $\Avar{u_i}$ are disjoint.
Hence, by combining such facts, we infer that $|\Avar{w_i}^{\str{M}}| \geq |\Avar{u_i}^{\str{M}}| \cdot |\Avar{v_i}^{\str{M}}|$.
But by~\cref{lemma:link-count-and-bfunc}(ii) we also know that every element from $\Avar{w_i}^{\str{M}}$ is connected via $\Mult{i}{}^{\str{M}}$ to some element from $\Avar{u_i}^{\str{M}}$. Hence the mentioned inequality becomes the equality.

\subsection{Proof of~\cref{thm:fo2undec}}

We need to show that $\varepsilon$ has a solution iff $\phireduction{\varepsilon}$ is finitely satisfiable.
For the ``if'' direction take any finite model $\str{M}$ of $\phireduction{\varepsilon}$ and let $S : \Var(\varepsilon)$ be a function that maps a variable $v$ to $|\Avar{v}^{\str{M}}|$.
From~\cref{lemma:multiplication},~\cref{lemma:addition},~\cref{fact:phiuequalone} it is easy to conclude that $S$ is indeed the solution of $\varepsilon$.

For the second direction assume that $\varepsilon$ has a solution and call it $S$. 
Let $\str{M}$ be any finite structure satisfying all the following conditions:
\begin{enumerate}[a)]
    \item $|M|$ is even and $\frac{|{M}|}{2} > max \{S(v)+S(u) : v, u  \in\Var(\varepsilon)\}$ and $ |M| > \sum_{i \in\Var(\varepsilon)} S(i) $
    \item All $a \in M$ satisfies at most one $\Avar{u}(x)$ for $u \in \Var(\varepsilon)$
    \item For all $v \in\Var(\varepsilon)$ we have $|\Avar{v}(x)|_x = S(v)$.
    \item For all entry $\varepsilon_i$ of the form $u_i + v_i = w_i$ or $u_i \cdot v_i = w_i$ all domain elements satisfying $\Avar{u_i}(x) \lor \Avar{v_i}(x)$ are labelled with $\FHalf{i}$ and all elements satisfying $\Avar{w_i}(x)$ are labelled with $\SHalf{i}$.
    \item For al $i \leq |\varepsilon|$ exactly half of domain elements are labelled with  $\FHalf{i}$ while the other half is labelled with $\SHalf{i}$.
    \end{enumerate}

    Additionally, for all $\varepsilon_i$ of the form $u_i \cdot v_i = w_i$ the structure $\str{M}$ should satisfy:
    \begin{enumerate}[A)]
    \item The relation $\Mult{i}{}^{\str{M}}$ is a subset of $\Avar{u_i}^{\str{M}} \times \Avar{w_i}^{\str{M}}$.
    \item For each $a \in \Avar{u_i}^{\str{M}}$, that the total number of $\Mult{i}{}^{\str{M}}$-successors of $a$ is equal to~$|\Avar{v_i}^{\str{M}}|$.
    \item Every element from $\Avar{w_i}^{\str{M}}$ has exactly one $\Mult{i}{}^{\str{M}}$-predecessor.
    \end{enumerate}

One can readily check, by routine case enumeration, that $\str{M} \models \phireduction{\varepsilon}$.
What remains to be done is to show that such a structure $\str{M}$ actually exists.
Hence, let $\str{M}$ be any finite structure satisfying~(a). It is obvious that such a structure exists. Due to the fact that $|M| > \sum_{i \in\Var(\varepsilon)} S(i)$ we can extend $\str{M}$ by interpret predicates $\Avar{u}$ (where $u \in \Var(\varepsilon)$) in $\str{M}$, such that (b) and (c) are satisfied. 

We next interpret the relational symbols $\FHalf{i}$ in $\str{M}$ such that the extended structure satisfy (d) and (e). It can be done due to the fact that $|\Avar{u_i}(x) \lor \Avar{v_i}(x)|_x = S(u_i) + S(v_i)$ and from the fact that $\frac{M}{2} > S(u_i) + S(v_i)$.  
Next we put $(\SHalf{i})^{\str{M}} = M \setminus (\FHalf{i})^{\str{M}}$.
Next, we  interpret $\Mult{i}{}$ in such a way that it satisfies all the remaining conditions. It is possible due to (a) and (b) and the fact that $|\Avar{w_i}|^{\str{M}} = |\Avar{u_i}^{\str{M}}| \cdot |\Avar{v_i}^{\str{M}}|$.
Since each of the steps is correct, the structure $\str{M}$ actually exists, which finishes the proof.

\subsection{More details on the proof of~\cref{thm:gf-local-undec}}
Let us recall that $\varphi_{\textit{func}}$ is defined as follows:
    \begin{align}
    \forall x \quad x=x \rightarrow & \\
    & (\forall y\quad F(x,y) \rightarrow R(x,y)) \quad \land \\
    & \existspercwithrel{=}{50}{R}y R(x,y) \land H(x,y) \quad\land \\
    & (\forall y\quad F(x,y) \rightarrow ( \neg H(x,y) \lor x=y)) \quad\land \\
    & (\existspercwithrel{=}{50}{R}y (R(x,y) \land ((H(x,y) \land x\neq y) \lor F(x,y)))
    \end{align}

Call the above lines (a), (b), (c), (d), (e). We prove two lemmas.

\begin{lemma}
For all models $\str{M}$ of $\varphi_{\textit{func}}$ we have that $F^{\str{M}}$ is functional.
\end{lemma}
\begin{proof}
In the proof we write $|\varphi(x,y)|_{R,y}$ means the total number of $y$ satisfying both $R(x,y)$ and $\varphi(x,y)$.
From (c) and (e) we have $|((H(x,y) \land x\neq y) \lor F(x,y)|_{R,y} = |H(x,y)|_{R,y}$.
By applying De Morgan's law in (d) we can see that for all $x$ the set of elements $y$ witnessing $F(x,y)$ and the set of elements $y$ satisfying $(H(x,y) \land x\neq y)$ are disjoint. Thus we know that the equation $|((H(x,y) \land x\neq y) \lor F(x,y)|_{R,y} = |H(x,y) \land x\neq y|_{R,y} + |F(x,y)|_{R,y}$ holds. 
By applying simple transformations we conclude $ |F(x,y)|_{R,y} = |H(x,y)|_{R,y} - |H(x,y) \land x\neq y|_{R,y}$.
We can see that $|H(x,y)|_{R,y} - |H(x,y) \land x\neq y|_{R,y}$ is equal $1$ if both $H(x,x)$ and $R(x,x)$ and $0$ otherwise.
Hence, we conclude that for all $x$ the number of $R$-successors $y$ of $x$ satisfying $F(x,y)$ is either zero or one. 
By (b) we know that $F^{\str{M}}$ is a subset of $R^{\str{M}}$ thus $F$ is indeed functional.
\end{proof}

\begin{lemma}
Every finite $\str{M}$ with a functional $F^{\str{M}}$ can be extended by interpretation of $R, H$ in a way that $\str{M} \models \varphi_{\textit{func}}$.
\end{lemma}
\begin{proof}
Take any such $\str{M}$ and for all element $a \in M$ we define $R_a$ and $H_a$ as follows.
\begin{itemize}
    \item If there is $b \neq a$ such that $(a,b) \in F^{\str{M}}$ we put $R_a = \{ (a,b), (a,a) \}$ and $H_a = \{ (a,a)\}$
    \item If $(a,a) \in F^{\str{M}}$ holds then we take any $b \neq a$ and put define $R_a, H_a$ as above.
    \item If there is no $b$ such that $(a,b) \in F^{\str{M}}$ we keep $R_a$ and $H_a$ empty.
\end{itemize}
Now put $R^{\str{M}} = \textstyle \bigcup_{a \in M} R_a$ and $H^{\str{M}} = \textstyle \bigcup_{a \in M} H_a$.
It is easy to verify that such an extended $\str{M}$ satisfies $\varphi_{\textit{func}}$.
\end{proof}


\section{Results that are transferable to general satisfiability}
\label{app:gen-sat}

\subparagraph*{Undecidability results.}

Since the semantics of global percentage quantifiers only make sense over finite domains, all our undecidability results involving global percentage quantifiers hold also for general satisfiability (\ie general = finite).
For our results concerning their local counterparts, namely Corollary~\ref{cor:local-percentage}, note that the relation $U^{\str{M}}$ is forced to be universal. 
Thus, by the fact that we consider models where the total number of $U^{\str{M}}$-successors of any node is finite, this implies that the whole domain is also finite. 
Hence, also in this case we have that general satisfiability and the finite one are the same.
Finally, the undecidability of $\GF$ with local percentage quantifiers over arbitrary finite-branching structures can be concluded by routinely checking that in the undecidability proof of Graedel~\cite{Gradel:undec} the constructed models are finite branching.

\subparagraph*{Decidability results.}
It is not difficult to see that the decision procedure described in Section~\ref{sec:decidable} 
also holds for the general satisfiability problem for $\GFtpres$.
Indeed, note that the general satisfiability problem for $\Ct$ is decidable~\cite{Pratt-Hartmann05}.
The only part that is different is the correctness proof when
the model $\str{M}^*\models \Psi^*$ is infinite.
In this case, to construct the model $\str{M}\models \Psi$,
we make infinitely many copies: $\str{M}_{i,j}$, where $i\in \{0,1,2\}$
and $j\in \N$.
The merging process to eliminate the periodic elements is the same.

We believe that a rather standard technique à la tableaux of constructing a tree model of $\GFtpres$ level-by-level can be employed here and yields $\ExpTime$ upper bound for the general satisfiability problem. 
We stress that this approach exploit the ``infiniteness'' of models in a rather heavy way: at any level of the infinite tree model we can always pick fresh witnesses in order to satisfy formulae and we do not need to worry that at some point the models must be rolled-up to form a finite structure. 
Since we are not ready with all the details at the time of submission, we delegate this result to the journal version of the paper.


\newcommand{\pathrho}{\rho}
\newcommand{\cycle}{\rho}
\newcommand{\girth}{\mathsf{girth}}

\newcommand{\role}[1]{\mathit{#1}}      
\newcommand{\rolep}{\role{p}}           
\newcommand{\roler}{\role{r}}           
\newcommand{\roles}{\role{s}}           
\newcommand{\rolet}{\role{t}}   
\newcommand{\concept}[1]{\mathrm{#1}}       
\newcommand{\conceptA}{\concept{A}}         
\newcommand{\conceptB}{\concept{B}}         
\newcommand{\conceptC}{\concept{C}}         
\newcommand{\conceptD}{\concept{D}}         
\newcommand{\queryatom}{\alpha}      
\newcommand{\query}[1]{\mathit{#1}}  
\newcommand{\queryq}{\query{q}}      
\newcommand{\match}[1]{#1}          
\newcommand{\matchpi}{\match{\pi}}  
\newcommand{\matcheta}{\match{\eta}}  
\newcommand{\modelsmatch}[1]{\models_{#1}} 

\newcommand{\queryVar}[1]{\mathrm{Var}{(#1)}}   
\newcommand{\queryVarq}{\queryVar{\queryq}}     

\newcommand{\modelsfin}{\models_\mathrm{fin}}       
\newcommand{\modelsoptfin}{\models_\mathrm{(fin)}}  
\newcommand{\var}[1]{\mathit{#1}}   
\newcommand{\varx}{\var{x}}         
\newcommand{\vary}{\var{y}}         
\newcommand{\varz}{\var{z}}         
\newcommand{\varv}{\var{v}}         
\newcommand{\varu}{\var{u}}         

\newcommand{\homo}[1]{\mathfrak{#1}}    
\newcommand{\homof}{\homo{f}}           
\newcommand{\homog}{\homo{g}}           
\newcommand{\homoh}{\homo{h}}           
\newcommand{\ishomoto}{\vartriangleleft} 
\newcommand{\homeq}{\rightleftarrows} 
\newcommand{\isoeq}{\cong} 

\section{Appendix for conjunctive query entailment}\label{appendix:querying}
 
In this section we formally prove~\cref{thm:result-on-query-answering}.
Before we start, let us revisit the basics on conjunctive query entailment problem over ontologies.

\subsection{Preliminaries on queries, homomorphisms and Gaifman graphs}
  \emph{Conjunctive queries} (CQs) are conjunctions of \emph{atoms} of the form \(\roler(\varx, \vary) \) or \(\conceptA(\varz) \), where \(\roler \) is a binary relational symbol, \(\conceptA \) is a unary relational symbol and \(\varx, \vary, \varz \) are variables from some countably infinite set of variable names.
  We denote with $|\queryq|$ the number of its atoms and with $\queryVar{\queryq}$ the set of all variables that appear in $\queryq$.

Let $\str{A}$ be a structure, $\queryq$ a CQ and $\matcheta: \queryVar{\queryq}\to A$ be a variable assignment.
  We write \(\str{A} \modelsmatch{\matcheta} \roler(\varx,\vary) \) if~\((\matcheta(\varx),\matcheta(\vary))\in \roler^\str{A} \) and~\(\str{A} \modelsmatch{\matcheta} \conceptA(\varz) \) if \(\matcheta(\varz) \in \conceptA^\str{A} \). 
  We say that~\(\matcheta \) is a \emph{match} for \( \str{A} \) and \(\queryq \) if \( \str{A} \modelsmatch{\matcheta} \alpha \) holds for every atom $\alpha \in \queryq$ and that~\( \str{A} \) \emph{satisfies} \(\queryq \) (denoted with: \(\str{A} \models \queryq \)) whenever \(\str{A} \modelsmatch{\matcheta} \queryq \) for some match \(\matcheta \). 
  The definitions are lifted to formulae: $\queryq$ is \emph{(finitely) entailed} by a $\GFtpres$ formula $\varphi$, written: $\varphi \models \queryq$ ($\varphi \modelsfin \queryq$) if every (finite) model of $\varphi$ satisfies $\queryq$.
  When \(\str{A} \models \queryq \) but \(\str{A} \not\models \queryq \), we call \(\str{A} \) a \emph{countermodel} for \(\varphi \) and \(\queryq \). Moreover, if such an $\str{A}$ is finite, we call it a \emph{finite countermodel}.
    The \emph{(finite) query entailment} problem for $\GFtpres$ is defined as follows: given a formula $\varphi$ and a CQ $\queryq$ verify if $\varphi$ (finitely) entails $\queryq$.

  Note that CQs may be seen as structure: for a query~$\queryq$ we define a structure $\str{Q}_{\queryq}$ satisfying $(\varx, \vary) \in \roler^{\str{Q}_{\queryq}}$ iff $\roler(\varx, \vary) \in \queryq$ and $\varx \in \conceptA^{\str{Q}_{\queryq}}$ iff $\conceptA(\varx) \in \queryq$.

  A \emph{homomorphism} $\homoh : \str{A} \to \str{B}$ is a function that maps every element of $A$ to some element from $B$ and it preserves unary and binary relations, \ie we have that $a \in \conceptA^{\str{A}}$ implies that $\homoh(a) \in \conceptA^{\str{B}}$ and $(a, b) \in \roler^{\str{A}}$ implies $(\homoh(a), \homoh(b)) \in \roler^{\str{B}}$ for all binary relational symbols $\roler$, unary relational symbols $\conceptA$ and elements $a, b \in A$.
    Since queries can be seen as structures, their matches can be seen as homomorphisms.

We employ the notion of Gaifman graphs.  
Intuitively, the \emph{Gaifman graphs} $G_{\str{A}}$ of a structure $\str{A}$ is the underlying undirected graph structure of $\str{A}$. 
More precisely the graph $G_{\str{A}} = (V_{\str{A}}, E_{\str{A}})$ is composed of nodes $V_{\str{A}} := A$ and edges $E_{\str{A}}$, for which $(a, b) \in E_{\str{A}}$ if $a \neq b$ and there is a binary relation name $\roler$ such that $(a, b) \in \roler^{\str{A}}$ or $(b,a) \in \roler^{\str{A}}$. 
The \emph{girth} of $G_{\str{A}}$ is the length of its shortest cycle or $\infty$ if $G_{\str{A}}$ does not have any cycles. 
The \emph{girth} of $\str{A}$ is the girth of its Gaifman Graph.\footnote{Self-loops in $\str{A}$ are not counted as cycles as $G_\str{A}$ is simple.}
A structure $\str{A}$ is tree-shaped if $G_{\str{A}}$ is a tree.

The notion of Gaifman graphs is adjusted to queries $\queryq$, \ie the \emph{query graph} $G_{\queryq}$ is simply the Gaifman graph of its corresponding structure. Similarly, we speak about tree-shaped queries as well as tree-shaped matches, \ie query matches whose induced substructures are tree-shaped.

\subsection{Reducing conjunctive query entailment to satisfiability}

Our goal is to reduce (finite) conjunctive query entailment problem  to (finite) satisfiability (in exponential time, which is optimal).
We rely on previous results by the first author that appeared in the workshop paper~\cite{BaaderBR19}.
Our proof methods rely on two facts:
\begin{enumerate}
    \item Tree-shaped query matches (\ie those whose underlying graph forms a tree, more precisely a graph of treewidth $1$) can be efficiently blocked.
    \item If there is a countermodel $\str{A}$ for a formula $\varphi$ and a query $\queryq$ then there is also a countermodel $\str{B}$ of high \emph{girth}, meaning that it is sufficiently tree-like for a query $\queryq$ to match only in a tree-shaped way. In other words, from a given structure we ``eliminate'' all non-tree-shaped query matches of $\queryq$. We achieve this by an appropriate model transformation, called \emph{pumping}, similar to the one recently introduced by the first author and his coauthors in the workshop paper~\cite{BaaderBR19}.
\end{enumerate}

\noindent We discuss them in separate sections.

\subsection{Eliminating tree-shaped query matches}

We start from the first point. 
A query $\queryq'$ is a \emph{treeification} of $\queryq$ if it is tree-shaped and can be obtained from $\queryq$ by selecting, possibly multiple times, two of its variables and identifying them.
The set $\mathsf{Tree}(\queryq)$, \ie the set of all its treeifications of $\queryq$, is clearly of at most exponential size in $|\queryq|$.  

To detect tree-shaped query matches we employ the well-known rolling-up technique~\cite{Tessaris01,GlimmLHS08} that for a given tree-shaped $\queryq$ produces a $\GFt$ defining an unary relation $\conceptC_{\queryq}$, whose interpretation is non-empty in $\str{A}$ iff there is a match of $\queryq$ in $\str{A}$.
Hence, by imposing that such unary predicates are empty in $\str{A}$ for all treeifications of $\queryq$, we can conclude that $\str{A}$ does not have any tree-shaped query matches.
Let $\preceq$ be a descendant ordering on variables of a tree-shaped~$\queryq$.
Then for any leaf $\varv$ of $G_\queryq = (V_\queryq, E_\queryq)$ (\ie the $\preceq$-maximal element of $\preceq$) we define a unary relation $\conceptC_{\queryq}^{\varv}$ as
\[
    \conceptC_{\queryq}^{\varv}(x) := \bigwedge_{\conceptA(\varv) \in \queryq} \conceptA(x) \land \bigwedge_{\roler(\varv,\varv) \in \queryq} \roler(x,x),
\]
while for non-leaf node $\varv$ the unary relation $\conceptC_{\queryq}^{v}$ is defined with as above but we additionally append:
\[
    \bigwedge_{ (\varv,\varu) \in E_{\queryq}, \varu \prec \varv }  \exists{y}\; \left( \bigwedge_{\roler(\varu,\varv) \in \queryq}\roler(\varx,\vary) \; \wedge \bigwedge_{\roler(\varv,\varu) \in \queryq}\roler(\vary, \varx) \right) \wedge \conceptC_{\queryq}^{\varu}(y)
\]
We put $\conceptC_{\queryq} := \conceptC_{\queryq}^{\varv_r}$, where $\varv_r$ is the root variable of $\queryq$ (\ie the $\preceq$-minimal variable). Note that the size of $\conceptC_{\queryq}$ is linear in $|\queryq|$.

\begin{fact}\label{prop:rolling-up}
For a tree-shaped query $\queryq$ we have $\str{A} \models \queryq$ if and only if $(\conceptC_{\queryq})^{\str{A}} \neq \emptyset$.
\end{fact}
\begin{proof}

We proceed by induction over $\preceq$, where the inductive assumption is that for all variables $\varu \prec \varv$ we have that $a \in (\conceptC_{\queryq}^{\varu})^{\str{A}}$ iff there is a homomorphism $\homoh$ from the subtree rooted at $\varu$ to $\str{A}$ with $\homoh(\varu) = a$.

\begin{itemize}
    \item Base case. We have that the following unary predicate is non-empty and contains $a$
\[
    \conceptC_{\queryq}^{\varv}(x) := \bigwedge_{\conceptA(\varv) \in \queryq} \conceptA(x) \land \bigwedge_{\roler(\varv,\varv) \in \queryq} \roler(x,x),
\]
    and it basically states that $a$ satisfies all atoms involving $a$ as a sole variable.
    Hence, $\homoh(\varv) := a$ is indeed a homomorphism. 
    For the opposite way, the existence of a homomorphism ensures us that all the conjuncts of $\conceptC_{\queryq}^{\varv}$ are satisfied at $a$.

    \item Take any variable $\varv$ and assume that there is a homomorphism $\homoh$ from the subtree rooted at $\varv$ to $\str{A}$ with $\homoh(\varv) = a$.
    We want to show that $a$ satisfies
\[
    \bigwedge_{ (\varv,\varu) \in E_{\queryq}, \varu \prec \varv }  \exists{y}\; \left( \bigwedge_{\roler(\varu,\varv) \in \queryq}\roler(\varx,\vary) \; \wedge \bigwedge_{\roler(\varv,\varu) \in \queryq}\roler(\vary, \varx) \right) \wedge \conceptC_{\queryq}^{\varu}(y)
\]
    , since the rest is in the base case.
    By the existence of a homomorphism we know that there are domain elements $\homoh(\varu) = a_\varu$.
    Since $\homoh$ is a homomorphism we know that if $\roler(\varv, \varu) \in \queryq$ then also $(a, a_\varu) \in \roler^{\str{A}}$ holds. Similarly for inverse roles. By the inductive assumption we also know that $a_\varu \in (\conceptC^{\varv}_{\queryq})^{\str{A}}$. Hence, $a \in (\conceptC^{\varv}_{\queryq})^{\str{A}}$.
    For the opposite direction the homomorphism is constructed by taking $\homoh(\varv) = a$ and then by selecting a witness $a_{\varu}$ for a formula
\[
    \bigwedge_{ (\varv,\varu) \in E_{\queryq}, \varu \prec \varv }  \exists{y}\; \left( \bigwedge_{\roler(\varu,\varv) \in \queryq}\roler(\varx,\vary) \; \wedge \bigwedge_{\roler(\varv,\varu) \in \queryq}\roler(\vary, \varx) \right) \wedge \conceptC_{\queryq}^{\varu}(y)
\]
     and putting $\homoh(\varu) := a_{\varu}$.
\end{itemize}
\end{proof}

The above fact can be justified by induction over $\preceq$, where the inductive assumption is that for all variables $\varu \prec \varv$ we have that $a \in (\conceptC_{\queryq}^{\varu})^{\str{A}}$ iff there is a homomorphism $\homoh$ from the subtree rooted at $\varu$ to $\str{A}$ with $\homoh(\varu) = a$.

\subsection{Obtaining counter-models of high girth}
\newcommand{\G}[0]{\mathsf{G}}
\newcommand{\two}[0]{\mathsf{2}}
\newcommand{\zeroS}[0]{\mathsf{zero}_{S}}
\newcommand{\EA}{\mathsf{E}_\str{A}}
\newcommand{\GA}{\mathsf{G}_\str{A}}
\newcommand{\VA}{\mathsf{V}_\str{A}}
\newcommand{\zeroEI}[0]{\mathsf{zero}_{\EA}}
\newcommand{\pump}[1]{\mathsf{pump}(#1)}
\newcommand{\pumpA}[0]{\pump{\str{A}}}
\newcommand{\GpumpA}[0]{\mathsf{G}_{\pumpA}}
\newcommand{\EpumpA}[0]{\mathsf{E}_{\pumpA}}

We start from auxiliary definitions.
For a given finite set~$S$ we denote with~$2^S$ the set of all boolean functions
with the domain~$S$. The unique function of the form~$S \rightarrow \{ 0 \}$ is denoted 
with~$\zeroS$. We drop the subscript~$S$, whenever it is known from the context.
Once two boolean functions~$f,g \in 2^S$ and an element~$s \in S$ are given, 
we say that $f$ is~\emph{nearly-$s$-equal}~$g$, 
symbolized by~$f \approx_s g$, if $f$ and~$g$ are equal on 
all arguments from~$S$ except for~$s$, for which~$g(s) = 1{-}f(s)$ holds.

In what follows, we will construct from $\str{A}$  a structure $\pumpA$ having the girth doubled.
The construction of $\pumpA$ is quite technical, so we provide some informal intuitions behind it.
We construct~$\pumpA$ from~$\str{A}$ by equipping each domain 
element from~$A$ with a set of coins (implemented as a boolean function), 
one coin per each edge from the initial structure. Then, we define connections 
between elements in the domain of $\pumpA$ in such a way that two elements~$u$ and~$v$ 
are connected only if they carry nearly the same sets of coins 
except that they differ on the side of unique coin responsible 
for the edge~$(u,v)$. Thus, in every round-trip, visited edges will appear even number of times 
(due to the fact that crossing an edge requires tossing the 
appropriate coin). Hence the girth of the obtained structure 
will be strictly greater than the girth of $\str{A}$.
\begin{definition} \label{def:pump}
Let~$\str{A}$ be a structure and let~$\GA=(\VA, \EA)$ 
be its Gaifman graph. We define the structure~$\pumpA$ with the domain $A \times 2^{\EA}$ as follows:
\begin{itemize}
\item For all elements $a \in A$ and all unary relational names $\conceptC$
we put $(a,f) \in \conceptC^{\pumpA}$ for all boolean function~$f \in 2^{\EA}$
whenever $a \in \conceptC^{\str{A}}$ holds.
\item For any binary relational symbols $\roler$ and any pair~$p=(v_1,v_2)$ of 
domain elements~$v_1 {=} (a_1,f_1)$ and~$v_2 {=} (a_2, f_2)$ with $a_1 {\neq} a_2$ 
we set $p \in \roler^{\pumpA}$ iff $(a_1, a_2) \in \roler^{\str{A}}$ and $f_1 \approx_{(a_1,a_2)} f_2$. 

\item For all domain elements $a$ and binary relational symbols $\roler$ s.t. $(a, a) \in \roler^{\str{A}}$ we put $((a, f), (a,f)) \in \roler^{\pumpA}$ for all~$f \in 2^{\EA}$.
\end{itemize} 
\end{definition}

We next aim to prove that the pumping method works as desired, \ie that it increases the girth and it is model-preserving.
We show the former goal first. The proof hinges upon the fact that in the projection of a cycle in $\pumpA$ some edge from $\str{A}$ is repeated twice.

\begin{lemma}\label{lemma:pumping-increases-girth}
For any structure~$\str{A}$ with a finite girth, the girth of~$\pumpA$ is 
strictly greater than the girth of~$\str{A}$.
Moreover, if $\str{A}$ is finite then $\pumpA$ is also finite.
\end{lemma}
\begin{proof}
Let~$\GA$ [resp.~$\GpumpA$] be the Gaifman graph of~$\str{A}$ [resp.~$\pumpA$].
Let~$\rho_{\pumpA} := 
(v_1, f_1) \rightarrow^{e_1} \ldots \rightarrow^{e_n} (v_{n+1}, f_{n+1}) {=} (v_1, f_1)$ 
be an arbitrary shortest cycle in~$\GpumpA$. From the construction of~$\pumpA$
we know that~$((v_i, f_i), (v_{i+1},f_{i+1})) \in \EpumpA$ holds 
iff~$(v_i, v_{i+1}) \in \EA$ holds. 
Thus the projection~$\rho_\str{A}$ of~$\rho_{\pumpA}$ onto~$\str{A}$ 
(\ie $\rho_\str{A} = v_1 \rightarrow^{e_1'} \ldots \rightarrow^{e_n'} v_n$) 
is a cycle in~$\str{A}$. To conclude that~$|\rho_\str{A}| > \girth(\str{A})$ it is 
sufficient to prove that some pair of elements~$(v_i, v_{i+1})$ appear 
at least twice in~$\rho_\str{A}$. 
Fix any edge~$e = (u,v)$ from~$\rho_\str{A}$.
W.l.o.g. assume that~$f_1(e) = 0$ and let~$i$ be the 
minimal index, such that~$((u, f_i), (v,f_{i+1})) \in \rho_{\pumpA}$. 
Hence~$f_i(e) = 0$, but~$f_{i+1}(e) = 1$ (from the construction of~$\pumpA$). 
If there would be no~$j > i$, s.t.~$((u, f_j), (v,f_{j+1})) \in \rho_{\pumpA}$ 
then~$f_{n+1}(e) = 1$ (since crossing the edge~$e$ is the only way to flip the 
value of~$f(e)$), which leads to a contradiction with~$f_1 = f_{n+1}$. 
Thus~$e$ appears at least twice on the cycle~$\rho_\str{A}$.
\end{proof}

We next prove the latter, \ie that the pumping method is model-preserving.
\begin{lemma}\label{lemma:pumping-preserves-modelhood}
For all $\GFtpres$ $\varphi$ we have that if $\str{A} \models \varphi$ then $\pumpA \models \varphi$.
\end{lemma}
\begin{proof}
We can focus on $\varphi$ in normal forms. First, note that by the construction $a$ satisfies the same unary relations as $(a, f)$ and has the same self-loops.
Thus all atomic concepts, their boolean combinations and ``self-loop'' concepts from $\varphi$ are satisfied in $(a, f)$ in $\pumpA$ iff there are satisfied by $a$ in $\str{A}$.
Second, observe that the ``neighbourhood'' of $a \in A$ is preserved: an element $(a, f)$ is connected to the elements from $\{ (b, f_{b}) \}$ with $a$ and $b$ originally connected in $\str{A}$ and $f \approx_{(a, b)} f_{b}$. Moreover, by the construction, $((a, f), (b, f_{b})) \in \roler^{\pumpA}$ iff $(a, b) \in \roler^\str{A}$. 
Hence, it follows that all two-types of all pairs of elements are not only identical but the degree of every node form $\str{A}$ is the same as the degree of its corresponding elements in $\pumpA$. Thus the satisfaction of Presburger constraints is preserved, finishing the proof.
\end{proof}

\subsection{Exponential reduction from querying to satisfiability}

We conclude the section by applying the pumping method to the (finite) conjunctive query entailment problem.
Take a $\GFtpres$ formula $\varphi$ and a CQ $\queryq$. 
Our reduction simply checks the (finite) satisfiability of $\varphi' = \varphi \land \bigwedge_{\queryq' \in \mathsf{Tree}(\queryq)} \forall{x} \neg\conceptC_{\queryq'}(x)$.
The following lemma claims the correctness:

\begin{lemma}\label{lemma:ziq-reduction}
For any $\GFtpres$ we have that $\varphi \models \queryq$ (in the finite) iff $\varphi'$ is (finitely) unsatisfiable.
\end{lemma}
\begin{proof}
We show the proof for both finite and unrestricted setting at the same.
For the first direction assume that $\varphi \models \queryq$ (in the finite) and ad absurdum assume that $\varphi'$ has a (finite) model $\str{A}$.
By applying the pumping method to $\str{A}$ at least $|\queryq|$ times, by~\cref{lemma:pumping-increases-girth} we obtain a (finite) $\str{B}$ of girth greater than $|\queryq|$. By~\cref{lemma:pumping-preserves-modelhood} we conclude that $\str{B} \models \varphi'$. Since $\varphi' \models \varphi$ we also know that $\str{B} \models \varphi$. Since $\varphi \models \queryq$ we infer that there is a tree-shaped match of $\queryq$ on $\str{B}$ due to the fact that the girth of $\str{B}$ is greater than the number of atoms of $\queryq$. Hence, there is a treeification $\queryq'$ of $\queryq$ that matches $\str{B}$ implying (by~\cref{prop:rolling-up}) that $\conceptC_{\queryq'}^{\str{B}} \neq \emptyset$. 
It contradicts the satisfaction of an extra conjuncts in $\varphi'$ enforcing that $\conceptC_{\queryq'}^{\str{B}} = \emptyset$.
For the second direction, by contraposition, it suffices to take a (finite) countermodel for $\varphi$ and $\queryq$, which by~\cref{prop:rolling-up} is also a model of $\varphi'$.
\end{proof}

Hence, by the decidability of finite and unrestricted satisfiability problem for $\GFtpres$ we conclude~\cref{thm:result-on-query-answering} and its unrestricted version.

\end{document}